\documentclass[a4paper,USenglish,cleveref,autoref,thm-restate]{lipics-v2021}

\hideLIPIcs  

\usepackage[camera]{dtrt}
\usepackage{comment}
\usepackage{cite}
\usepackage{xurl}
\usepackage{booktabs} 
\usepackage{float}
\usepackage{caption}
\usepackage{makecell}
\usepackage{cleveref}
\usepackage{siunitx}
\crefname{appendix}{appendix}{appendices}
\Crefname{appendix}{Appendix}{Appendices}

\usepackage{algorithm2e}
\usepackage{algorithmic}

\title{To Wait or To Probe: Arbitrage Competition on High-Throughput Blockchains}

\titlerunning{To Wait or To Probe: Arbitrage Competition on High-Throughput Blockchains}  

\author{Fei Wu\textsuperscript{*}}{King's College London, United Kingdom \and \url{https://www.mikuwill.me/} }{fei.wu@kcl.ac.uk}{https://orcid.org/0009-0004-5717-0219}{}

\author{Burak {\"O}z}{Flashbots, Germany}{burak@flashbots.net}{https://orcid.org/0009-0003-7508-7112}{}

\authorrunning{Fei Wu and Burak {\"O}z} 

\Copyright{Fei Wu and Burak {\"O}z} 

\ccsdesc[500]{Applied computing~Electronic commerce}

\keywords{Maximal Extractable Value, Cyclic Arbitrage, Spam, Base} 

\supplementdetails[linktext={\url{https://github.com/M1kuW1ll/base_arbitrage_competition}}]{Source Code}{https://github.com/M1kuW1ll/base_arbitrage_competition} 

\nolinenumbers 

\EventEditors{John Q. Open and Joan R. Access}
\EventNoEds{2}
\EventLongTitle{42nd Conference on Very Important Topics (CVIT 2016)}
\EventShortTitle{CVIT 2016}
\EventAcronym{CVIT}
\EventYear{2016}
\EventDate{December 24--27, 2016}
\EventLocation{Little Whinging, United Kingdom}
\EventLogo{}
\SeriesVolume{42}
\ArticleNo{23}

\newcommand{\added}[1]{\textcolor{blue}{#1}}
\newcommand{\deleted}[1]{\textcolor{red}{\sout{#1}}}

\renewcommand{\added}[1]{#1}
\renewcommand{\deleted}[1]{}

\begin{document}

\maketitle

\let\oldthefootnote\thefootnote
\let\thefootnote\relax
\footnotetext{\textsuperscript{*}Fei Wu performed work in part during an internship at Flashbots in Spring 2026.}
\let\thefootnote\oldthefootnote

\begin{abstract}
Maximal Extractable Value (MEV) on high-throughput blockchains can be captured through \emph{targeted search}, where bots identify opportunities off-chain and submit route-committed transactions, or through \emph{probabilistic search}, where bots submit repeated attempts that discover opportunities during on-chain execution. 
This distinction has direct implications for spam, blockspace consumption, and protocol revenue. We model how ordering granularity, fee floors, and opportunity-access shocks shape competition between these architectures.
Using cyclic arbitrage data on Base from June 2025 to February 2026, we develop a trace-level classifier for search architectures and show that the resulting labels correspond to distinct execution behavior. In our sample, probabilistic search accounts for only 23\% of arbitrage activity but produces 95\% of spam and consumes 20\% of Base gas. We test the model across three episodes: Flashblocks selects against broad on-chain probabilistic scanners; token-launch opportunity shocks temporarily revive probabilistic search; and higher fee floors select against probabilistic bots whose opportunity flow cannot sustain repeated attempts. After Base's configuration changes, protocol revenue shifts toward successful arbitrages and away from spam, probabilistic bots pay higher priority fees, and spam consumes a smaller share of blockspace.
\end{abstract}

\section{Introduction}
\label{sec:intro}
A book collector offers \$500 for a first-edition copy of \emph{The Hobbit}. A local bookstore will receive the book at an unknown time and sell it for \$400, with customers allowed to add a gratuity for priority service. The book becomes purchasable as soon as it enters the bookstore's inventory, but the public website updates only every ten minutes. Alice waits for the public update, then submits a purchase offer with a high gratuity. Bob repeatedly queries live inventory and buys if one check reveals the book first, imposing load on the store's system. Who wins the arbitrage? What would the bookstore prefer: targeted purchase attempts with higher gratuities, or repeated checks that might reveal opportunities earlier but consume system resources? Would the answer change if updates became more frequent or repeated checks became more costly?

This toy example captures a central design problem for Maximal Extractable Value (MEV) search on high-throughput blockchains: should searchers \emph{wait} for information and submit targeted execution attempts, or \emph{probe} the system with transactions whose profitability is resolved during execution? On Ethereum, MEV strategies are typically associated with targeted search: searchers identify and simulate an opportunity \emph{off-chain}, then submit a transaction committed to a specific execution route~\cite{flashboys,highfreq,qin2022quantifying,defiliquidations,nonatomic,wu2025measuringcexdex}. On high-throughput blockchains, fast block production, low fees, and limited visibility into pending transactions can make a second architecture viable: searchers submit many transactions that read live state and resolve opportunity discovery \emph{on-chain}, proceeding to further execution only when a profitable route is found~\cite{ozan_lioba_paper, fbspampost, wang2026blockspacepressureanalysisspam, mazorra2026timinggamesprobabilisticbackrunning,chen2023mevonsolana}.

We refer to these two architectures as \emph{targeted search} and \emph{probabilistic search}. Because opportunity discovery occurs at different stages, the two architectures differ in their sensitivity to ordering rules and fees, and therefore in their attempt intensity, blockspace consumption, and contribution to protocol revenue. Probabilistic searchers may submit many attempts per opportunity, most of which fail without producing useful state changes, a category referred to as \emph{spam}. These attempts consume blockspace and pay fees regardless of whether they succeed, so they affect both congestion and protocol revenue. Targeted searchers, by contrast, are expected to submit one transaction per opportunity and achieve higher success rates, while paying higher fees per attempt; however, they can also fail when an opportunity disappears before execution. The design question is therefore which architecture a protocol encourages, and how the resulting bot composition shapes blockspace consumption and protocol revenue.

This question has become increasingly relevant as high-throughput chains experiment with ordering rules and fee designs to reduce spam externalities and improve protocol revenue~\cite{base2026flashblocks, optimism2026stakebasedpriority, arbitrum2026timeboostintro, base2026minimumfeeincrease}. However, the equilibrium effect of such changes is not obvious. Existing work characterizes probabilistic search timing ~\cite{mazorra2026timinggamesprobabilisticbackrunning} and models equilibrium spam volume as a function of block capacity, gas fee floor, and fee-based ordering~\cite{wang2026blockspacepressureanalysisspam}, but typically focuses on a single search architecture. In this paper, we formalize the distinction between targeted and probabilistic search and study how protocol design shapes competition between them.


We develop a reduced-form equilibrium model of architecture competition in which ordering granularity, fee floors, and opportunity access shape the relative viability of targeted and probabilistic search (cf. \Cref{sec:model}). The model highlights two margins through which protocol changes affect spam: which architectures remain active, and how intensively active bots submit attempts.

We test these predictions on Base, which has documented spam activity \cite{fbspampost, ozan_lioba_paper, wang2026blockspacepressureanalysisspam} and underwent major configuration changes during our June 2025 to February 2026 sample period \cite{baseconfigchange}. We identify cyclic arbitrages and associated spam, and develop a trace-level classifier that distinguishes search architectures by the extent to which route discovery is observable on-chain before swap execution in successful arbitrages (cf. \Cref{sec:methodology}).

Our empirical analysis proceeds in three steps. We first validate that the classifier separates economically meaningful search architectures by showing that they differ in execution behavior, attempt efficiency, and spam production (cf. \Cref{sec:architecture_character}). We next study how arbitrage competition and the architecture relative viability evolve around three episodes: Flashblocks, the September surge, and the minimum-base-fee escalation (cf. \Cref{sec:flashblocks}-\ref{sec:fee_escalation}). Finally, we aggregate bot-level responses to chain-level outcomes, including priority-fee revenue and blockspace consumption (cf. \Cref{sec:chain_efficiency}).
We summarize our contributions as follows:
\begin{enumerate}
\item We formalize the distinction between \emph{targeted} and \emph{probabilistic} search architectures and develop a reduced-form equilibrium model in which ordering granularity, fee floors, and opportunity access shape the relative viability of the two architectures.
\item We develop a trace-level classifier that distinguishes the architectures through observable route discovery within successful arbitrages, and validate that the labels correspond to coherent transaction designs, spam behavior, and efficiency profiles.
\item We test the model's comparative statics on Base. Flashblocks selects against scan-intensive probabilistic bots and toward leaner survivors; a temporary opportunity shock revives probabilistic activity through short-lived entrants; and fee floor escalation selects against high-fee-exposure bots whose opportunity flow no longer supports repeated attempts.
\item We quantify the chain-level consequences. After configuration changes, priority-fee revenue shifts away from failed attempts toward realized arbitrages and higher-bidding probabilistic search attempts, while spam consumes a smaller share of blockspace.

\end{enumerate}

\subsection{Related work}

MEV competition by targeted search has been documented on Ethereum Layer-1~\cite{flashboys,wu2025measuringcexdex,qin2022quantifying,wang2022cyclicarbitragedecentralizedexchanges,nonatomic,torres2021frontrunnerjones,mchlaughlin2023largescalearbs, highfreq,defiliquidations}, alternate Layer-1s~\cite{oz2024playingmevgamefirstcomefirstserved, wang2026mevbinancebuilder}, Layer-2 networks~\cite{torres2025rollingshadowsanalyzingextraction}, and cross-network~\cite{burakcrosschainarb,obadia2021unity,gogol2024crossrollupmevnonatomicarbitrage}. Recent work documents probabilistic search on high-throughput blockchains: \cite{chen2023mevonsolana} introduces it as a viable strategy on Solana under faster block times and lower fees, and \cite{ozan_lioba_paper,fbspampost} show that similar strategies consume substantial blockspace on Base and Optimism, though \cite{ozan_lioba_paper} notes a different pattern on Arbitrum without formally distinguishing the two architectures. A related strand studies reverts and mechanisms for reducing failed MEV attempts~\cite{gogol2025priorityfailsrevertbasedmev, zhu2025quantifyingvaluerevertprotection,christof_timeboost}. Beyond empirical analysis, \cite{wang2026blockspacepressureanalysisspam} derives equilibrium spam volume under varying block capacity, minimum gas price, and priority-fee ordering, while \cite{mazorra2026timinggamesprobabilisticbackrunning} analyzes the timing game underlying probabilistic search competition. These models focus on probabilistic search in isolation. Our paper distinguishes targeted and probabilistic search architectures, connects failed attempts to architecture choice, and studies how ordering granularity, fee floors, and opportunity access shift competition between them.






\section{Preliminaries: Arbitrage Search Architectures}\label{sec:preliminaries}

\Cref{sec:intro} introduced targeted and probabilistic search as two economic architectures. A central setting for our analysis is cyclic arbitrage, where a bot routes a token through two or more decentralized exchange (DEX) pools and returns to the starting token with a positive net profit. On high-throughput chains, bots often submit candidate transactions before knowing with certainty whether execution will be profitable. Failed attempts can therefore arise under both architectures: a route-committed transaction may find that the targeted opportunity has disappeared, while a probabilistic search transaction may fail to find any profitable route during execution.
We use \emph{failed attempts} as an umbrella term for arbitrage attempts that do not execute a profitable swap sequence. Within this category, \emph{probes} terminate after state reads without executing swaps, while \emph{reverts} revert during execution. Both consume execution resources and are commonly treated as spam.
We next describe how these architectures appear at the \emph{implementation level}, focusing on where route discovery occurs and what evidence is visible in transaction traces.

\parhead{Off-chain discovery with direct execution.} 
The cleanest implementation of targeted search is a transaction that proceeds directly from entry to swap execution with no pool-state read before the first swap.\footnote{Consider example tx hash: \href{https://basescan.org/tx/0x1d93f62a009974da3ae575e7c623d88f09693b3e4e64ff0eabbac821f3fa091e}{\texttt{0x1d93f62a009974da3ae575e7c623d88f09693b3e4e64ff0eabbac821f3fa091e}}.} The route, DEX venues, and input amount are embedded in the transaction calldata, indicating that the opportunity was discovered, simulated, and parameterized off-chain before submission. 



\parhead{Route-confined on-chain evaluation.}

In the second implementation, the bot also begins with a pre-selected route embedded in the calldata, but reads live state from the relevant pools before committing to swaps.\footnote{Consider example tx hash: \href{https://basescan.org/tx/0x47a91fe3b9c31818f6c7f6699f6675297a8b926af48084e0c17cf8504a44df14}{\texttt{0x47a91fe3b9c31818f6c7f6699f6675297a8b926af48084e0c17cf8504a44df14}.}} These reads can be interpreted as a profitability check: if the route remains profitable, the transaction proceeds to swap execution; otherwise, it terminates early or reverts. Importantly, the pre-swap reads are confined to the venues that the transaction later executes against. This pattern is economically closer to targeted search than to probabilistic search, but it is intrinsically ambiguous from traces alone: it could also reflect discovery over a confined venue set that happens to coincide with the executed venues.


\parhead{Broad on-chain discovery.} The canonical implementation of probabilistic search is a transaction that scans a wider set of pools than it ultimately executes on.\footnote{Consider example tx hash: \href{https://basescan.org/tx/0xc020e534c39e203e9f3a42ea5f89ca023e486c38e15ca55e2a183ead1f4c885d}{\texttt{0xc020e534c39e203e9f3a42ea5f89ca023e486c38e15ca55e2a183ead1f4c885d}.}} The bot reads state from candidate venues to identify which subset offers a profitable route, then executes the selected route, or terminates or reverts if no profitable route is found. This strict superset relation is the trace-level signature of route discovery during execution.




\parhead{How bots read pool state.} We identify four mechanisms by which arbitrage bots read the DEX pool state during transaction execution.
\begin{enumerate}
    \item \textbf{Direct pool reads.} The bot, or a proxy contract it invokes via \texttt{CALL} or \texttt{DELEGATECALL}, issues a \texttt{STATICCALL} a public view function on the DEX pool contract --- for example, \texttt{getReserves()} on a Uniswap V2-style pool, \texttt{slot0()}, \texttt{liquidity()}, or \texttt{tickBitmap()} on a concentrated liquidity pool. The returned values describe the pool's current configuration and allow the bot to price candidate trades.
    \item \textbf{Balance-based reads.} For pools whose reserves correspond to ERC-20 token balances, the bot can issue a \texttt{STATICCALL} to \texttt{balanceOf(pool\_address)} on the relevant token contract to infer pool composition without interacting with the pool itself.\footnote{Consider example tx hash: \href{https://basescan.org/tx/0x2edd5bba978fbc4ded358f6f54f20643389c2dd02fd389ba55d06ab98e5ce358}{\texttt{0x2edd5bba978fbc4ded358f6f54f20643389c2dd02fd389ba55d06ab98e5ce358}.}}
    \item \textbf{Revert-based quotes.} The bot calls a quote function to the pool that simulates a swap and intentionally reverts with the quoted output amount encoded in the revert data. This pattern appears most often on Balancer pools: Balancer exposes dedicated query functions (\texttt{queryBatchSwap} on V2 \cite{balancerV2}, \texttt{quoteAndRevert} on V3 \cite{balancerv3}) that simulate the swap and return its result, and the protocol documents these as intended primarily for off-chain simulation via \texttt{eth\_call}.\footnote{Consider example tx hash: \href{https://basescan.org/tx/0x2ec4dbff9f3dd9f62ae648e32c49300283f742ef993b984fe2a92ec84ea24a22}{\texttt{0x2ec4dbff9f3dd9f62ae648e32c49300283f742ef993b984fe2a92ec84ea24a22}.}}
    \item \textbf{Helper contracts.} For some pool designs, the bot can read state by issuing a \texttt{STATICCALL} to a separate helper contract instead of the pool itself. For example, Uniswap V4 provides a \texttt{StateView} helper contract for off-chain simulation \cite{uniswapv4stateview}.\footnote{Consider example tx hash: \href{https://basescan.org/tx/0x11d2d594b738959a7be3ddb8bbc00f302e8d1414448c9821ed0e6e212e71e1dd}{\texttt{0x11d2d594b738959a7be3ddb8bbc00f302e8d1414448c9821ed0e6e212e71e1dd}.}}
\end{enumerate}
\section{A reduced-form equilibrium model}
\label{sec:model}

This section presents a stylized equilibrium model of how ordering granularity, transaction costs, and opportunity access shape the relative competitiveness of arbitrage bot architectures. 
Entry adjusts until expected profits are driven to zero. 

\parhead{Ordering regimes and bot architecture.} 
We consider two ordering regimes $r \in \{B, F\}$, where $B$ denotes the baseline regime and $F$ denotes a finer ordering regime. The relevant ordering unit is a regular full block in regime $B$, but a smaller sub-block slice in regime $F$.

We next describe two representative bot architectures. Following the implementation-level vocabulary of \Cref{sec:preliminaries}, we refer to them as \emph{off-chain discovery}, denoted by $O$, and \emph{on-chain discovery}, denoted by $C$, corresponding to targeted and probabilistic search, respectively. Bots compete over candidate arbitrage opportunities. We assume that under the chain's ordering regime, an opportunity is exhausted by the first valid claimant in the earliest profitable ordering unit.

\parhead{Same-slice competitiveness.}

On-chain discovery transactions perform state reads to evaluate candidate routes, so they consume more blockspace than route-committed off-chain transactions targeting the same opportunity. Under finer ordering, early sub-block slices have smaller effective capacity, so heavier on-chain transactions are more likely to be displaced, allowing a lean off-chain transaction to capture the opportunity in the earlier slice.
\footnote{Under Flashblocks on Base \cite{flashblocksconfig}, transactions exceeding 14M gas cannot fit in the first Flashblock.}



We summarize this disadvantage by a reduced-form \emph{same-slice competitiveness} function
\[
\eta_r \in [0,1),
\qquad r\in\{B,F\},
\]
where $\eta_r$ is the probability that an on-chain discovery transaction remains competitive in the same earliest profitable slice as an off-chain discovery transaction in regime $r$. The key restriction is $\eta_F < \eta_B$: on-chain discovery is less likely to remain competitive in the same profitable slice under a finer ordering regime than under the baseline regime.

\parhead{Opportunity access and free-entry payoffs.}
\added{The two architectures differ in opportunity access. An off-chain discovery transaction may encode several pre-identified routes, but its candidate route set is fixed before submission. An on-chain discovery transaction instead reads the current state and can select among a state-contingent set of candidate routes during execution. We summarize this potential breadth advantage through architecture-specific opportunity-access rates. Let $\lambda_i$ denote the expected number of candidate opportunities per model period accessible to architecture $i \in \{O,C\}$. With average gross value $V$ conditional on successful capture, $\lambda_iV$ is the gross opportunity value flow accessible to architecture $i$, before competition and transaction costs.\footnote{We assume $V$ to be common across architectures. In reality, off-chain discovery may extract more value through optimized off-chain simulation; this asymmetry would only reinforce the qualitative results.} We allow $\lambda_C \ge \lambda_O$ to represent the potentially broader access of on-chain discovery.}

\added{
Let $m_i$ denote the effective mass of independent competing strategies using architecture $i$. Here, ``mass'' is the standard continuum-entry measure of participation rather than a literal count of unique addresses or contracts observed on-chain. In a fraction $\eta_r$ of competitive states, an on-chain discovery transaction remains in the earliest profitable ordering slices and competes with off-chain discovery transactions. Conditional on this case, the opportunity is allocated uniformly across the total active mass $m_O +m_C$. In the remaining states, the on-chain discovery transaction is displaced, and only off-chain discovery bots compete, uniformly across $m_O$. Each bot incurs an expected cost $K_i(b)>0$ per modeled period, where $b>0$ is the minimum base fee. Applying this reduced form of the sharing rule to each architecture’s accessible opportunity flow yields the payoffs below.}

The expected payoff of an off-chain discovery bot in regime $r$ is 
\[
\Pi_O^r(m_O,m_C;b)
=
\eta_r\,\frac{\lambda_O V}{m_O+m_C}
+
(1-\eta_r)\,\frac{\lambda_O V}{m_O}
-
K_O(b).
\]
\added{The first term is the off-chain bot’s share when both architectures remain competitive in the earliest profitable slice; the second is its share when on-chain discovery is displaced; the final term is its expected cost. An on-chain-discovery bot earns opportunity value only in the fraction $\eta_r$ of states in which its heavier transaction remains competitive in the earliest profitable slice:}
\[
\Pi_C^r(m_O,m_C;b)
=
\eta_r\,\frac{\lambda_C V}{m_O+m_C}
-
K_C(b).
\]

A competitive equilibrium in regime $r$ is a pair $(m_O^{\ast,r},m_C^{\ast,r})$ such that no bot architecture earns positive expected profit, and any active bot architecture earns zero expected profit:
\[
\Pi_i^r(m_O^{\ast,r},m_C^{\ast,r};b)\le 0
\quad\text{for } i\in\{O,C\},
\quad
m_i^{\ast,r}>0
\;\Rightarrow\;
\Pi_i^r(m_O^{\ast,r},m_C^{\ast,r};b)=0.
\]
Thus, active architectures earn zero profit, and inactive architectures cannot enter profitably.

\begin{remark}[Blockspace competitiveness versus opportunity creation]
The model isolates the blockspace competitiveness channel rather than the opportunity-creation channel. Finer ordering may also change the value or arrival rate of arbitrage opportunities by reducing the order flow aggregated into each interval. In our notation, such effects would enter through $V$ or $\lambda_i$.
We hold $V$ and $\lambda_i$ fixed to focus on architecture composition. 
\end{remark}
For simplicity, we next focus on equilibria with active off-chain discovery architecture. Other equilibrium configurations can be derived analogously.

\begin{proposition}[Within-regime equilibrium]
\label{prop:regime_eq}
Fix $b\ge 0$, $r\in\{B,F\}$, and $\kappa_C(b)>\kappa_O(b)>0$. On the equilibrium branch with active off-chain discovery, the unique equilibrium is
\[
(m_O^{\ast,r},m_C^{\ast,r})
=
\begin{cases}
\left(
\dfrac{(1-\eta_r)\kappa_O(b)}{1-\kappa_O(b)/\kappa_C(b)},
\dfrac{\eta_r\kappa_C(b)-\kappa_O(b)}{1-\kappa_O(b)/\kappa_C(b)}
\right),
& \eta_r\kappa_C(b)>\kappa_O(b),\\[1.2em]
\left(\kappa_O(b),0\right),
& \eta_r\kappa_C(b)\le \kappa_O(b).
\end{cases}
\]
Thus, on-chain discovery is active if and only if $\eta_r\kappa_C(b)>\kappa_O(b)$, \added{i.e., its access-adjusted value-cost advantage is large enough to overcome its ordering disadvantage.}
\end{proposition}

Next, we show how a decline in same-slice competitiveness, from $\eta_B$ to $\eta_F$, shifts equilibrium participation away from on-chain discovery and toward off-chain discovery.

\begin{theorem}[Finer ordering reduces the competitiveness of on-chain discovery]
\label{thm:fb_competitiveness}
Fix $b \geq 0$, and suppose $\eta_F < \eta_B$ and both regimes are on the coexistence branch. Then
\[
m_C^{\ast,F}(b) < m_C^{\ast,B}(b),
\qquad
m_O^{\ast,F}(b) > m_O^{\ast,B}(b).
\]
Switching to the finer ordering regime reduces on-chain participation and increases off-chain participation on the coexistence branch.
\end{theorem}

To capture selection within on-chain discovery, we allow subtypes to differ in their \emph{per-attempt scan footprint}, which measures how much venue-search work a single attempt performs. It is distinct from \emph{attempt intensity}, i.e., how many attempts a bot sends. A larger scan footprint may broaden the candidate set evaluated within an attempt, but also makes the attempt heavier, reducing same-slice competitiveness under finer ordering and raising gas cost. This trade-off creates a selection margin within the on-chain discovery category.

\begin{proposition}[Finer ordering selects against high-footprint on-chain discovery]
\label{prop:selection_on_complexity}
Under the regularity conditions defined in \Cref{app:model_fb}, let on-chain discovery subtypes be indexed by a per-attempt scan-footprint parameter $x$, and let $x_r^\ast(b)$ denote the largest viable footprint in regime $r$. Then $x_F^\ast(b) < x_B^\ast(b)$,
so finer ordering selects against more scan-intensive on-chain-discovery strategies.
\end{proposition}




This selection result implies that, among surviving on-chain discovery bots, observable scan-footprint measures, such as the number of pre-swap state reads, should weakly fall under finer ordering. We formalize this prediction in \Cref{app:model_fb}, \Cref{cor:simpler_survivors}.




\begin{remark}[Off-chain discovery reliability under finer ordering]
\label{rem:offchain_reliability}
Finer ordering also shortens the interval between opportunity observation and execution, reducing the stale-state risk faced by off-chain discovery. In the model, this channel would raise the effective value-cost ratio of off-chain discovery by reducing failed attempts.
This stale-state channel and the blockspace-competitiveness channel both predict a shift away from broad on-chain discovery toward route-committed execution, so the composition test in \Cref{sec:flashblocks} cannot separately identify them. However, the stale-state channel more directly predicts improved off-chain attempt efficiency. The rise in off-chain discovery success rates in \Cref{tab:flashblocks_attempt_efficiency_spam} is therefore suggestive evidence for this margin, although not separately identified from changes in blockspace competition.
\end{remark}

The theorems above isolate the blockspace-competitiveness channel under a finer ordering regime.
We then turn to the cost channel. A higher minimum base fee raises the cost of every attempt, so it selects against architectures with higher fee exposure.



\begin{proposition}[A higher minimum base fee can disproportionately hurt on-chain discovery]
\label{prop:fee_floor}
Define architecture $i$'s fee exposure and on-chain discovery effective advantage as
\[
\chi_i(b):=\frac{\partial \log K_i(b)}{\partial b},
\qquad
\rho(b):=\frac{\kappa_C(b)}{\kappa_O(b)}.
\]
If $\chi_C(b)>\chi_O(b)$, then $\rho(b)$ is strictly decreasing in $b$. Equivalently, the coexistence condition $\eta_r\kappa_C(b)>\kappa_O(b)$ becomes harder to satisfy as the fee floor rises.
\end{proposition}

The compositional prediction requires asymmetric fee exposure. If the fee floor scaled every architecture's cost by the same proportional factor, then $\chi_C(b)=\chi_O(b)$, and a higher fee floor would reduce total entry without changing architecture composition. Empirically, this condition is consistent with on-chain discovery consuming more gas per attempt and sending more attempts per successful arbitrage (cf. \Cref{sec:architecture_character}).

The same logic applies within on-chain discovery: holding opportunity access and same-slice competitiveness fixed, higher fee floors reduce the advantage of on-chain subtypes with more fee-sensitive cost functions. We formalize this in \Cref{app:model_fee}, \Cref{prop:fee_exposure_selection}.


\begin{remark}[What fee exposure means empirically]
\label{rem:fee_exposure_decomposition}
Fee exposure depends not only on per-attempt scan footprint, but also on attempt intensity. Thus, higher fee floors should be interpreted as raising the hurdle rate for repeated probing, not as mechanically eliminating the most scan-intensive bots. High-exposure bots may remain viable if their opportunity flow is valuable enough, while lower-exposure bots may exit if their opportunity flow is weak. 
\end{remark} 

The blockspace and cost channels above hold the opportunity environment fixed. 
Opportunity access shocks move the other margin. Token launches and fresh DEX deployments can temporarily raise $\lambda_C$ by creating routes that off-chain discovery architectures have not yet indexed or parameterized. Such shocks can revive on-chain discovery participation even under otherwise unfavorable policy conditions.

\begin{proposition}[Opportunity-access shocks can revive on-chain discovery]
\label{prop:launch}
On the coexistence branch, $m_C^{\ast,r}$ is strictly increasing and $m_O^{\ast,r}$ strictly decreasing in $\lambda_C$. There exists a viability threshold $\bar{\lambda}_C(r,b)$ such that on-chain discovery is active iff $\lambda_C > \bar{\lambda}_C(r,b)$, and this threshold is strictly increasing in policy frictions. Hence, a temporary increase in $\lambda_C$ can revive on-chain discovery even when the policy conditions render it inactive.
\end{proposition}

\begin{remark}[Participation and observed spam]
\label{rem:participation_spam}
Free entry links the model's equilibrium masses to empirical architecture prominence: when a policy change lowers an architecture's expected payoff, the equilibrium supports
a smaller mass of participants using that architecture.
Empirically, lower architecture mass can be reflected in fewer active bots or a lower architecture share of arbitrage count or volume. Aggregate spam, however, also depends on attempt intensity per bot. Thus, spam may not fall proportionally if surviving or entering bots submit more attempts.
\end{remark}

\section{Data and Methodology}
\label{sec:methodology}

We use Base as a case study to measure arbitrage competition between search architectures and their response to protocol changes. Base is a natural setting because it has substantial DEX activity and spam volume \cite{ozan_lioba_paper,wang2026blockspacepressureanalysisspam, fbspampost}, and because it implemented major configuration changes, most notably Flashblocks and subsequent increases in the minimum base fee \cite{baseconfigchange}. This section describes the data construction and bot-architecture classification methodology.

\subsection{Data collection}
We collect cyclic arbitrages and associated spam activity on Base from June 1, 2025, through February 28, 2026, covering the periods of major configuration changes on Base \cite{baseconfigchange}.

\parhead{Successful cyclic arbitrages.}
We extend the cyclic arbitrage detection method of \cite{ozan_lioba_paper}. Using a custom Dune query, we detect 21,374,434 cyclic arbitrages by 4,365 unique arbitrage bot addresses over the observed period. For each transaction, we collect gas usage, fees, USD notional volume, and the realized arbitrage route, including traded tokens and DEX pools.

\parhead{Failed arbitrage attempts (spam).} 
We separately query reverted and probe transactions that are attributed to the same arbitrage bot addresses. Importantly, the bot does not need to interact with the DEX pool contract to read its state (cf. \Cref{sec:preliminaries}). 
To capture this broader class of probing behavior, we classify as probes successful transactions that emit no logs from external contracts other than the bot itself and carry ETH value $\le$ 1 WEI.\footnote{Simple ETH transfer also emits no external logs. Certain arbitrage bots typically attach 1 WEI to their transactions, such as \href{https://basescan.org/address/0xA4B9eDB99806bFB8BbEE89D13F4c296DAF225dEE}{\texttt{0xA4B9eDB99806bFB8BbEE89D13F4c296DAF225dEE}}. Thus, we set 1 WEI as a buffer.} 

This yields 985,559,360 spam transactions from these cyclic arbitrage bots over the observed period, including 35,653,523 reverts and 949,905,837 probes. Because transaction-level analysis of this full spam set is infeasible, we aggregate it at the bot-week level: for each bot in each week, we collect the count of probes and reverts, corresponding total gas usage and fees, and share of transaction count and gas usage out of the chain's throughput.

\subsection{Bot architecture identification}
\parhead{Transaction-level labeling.}
Because failed-attempt outcomes can arise from any of the three implementation patterns, distinguishing architectures requires trace-level evidence rather than transaction outcomes alone. We identify architecture from successful arbitrage transactions, where the executed route is observable.
The objective is to distinguish whether a transaction arrives with a route already committed, evaluates only the route it eventually executes, or scans multiple candidate venues before choosing an execution route.

Since tracing all 21.37M arbitrages is infeasible, for each bot in each week during the observed period, we randomly sample up to 50 arbitrage transactions, yielding 359,752 sampled transactions across 4,365 bots and 39 weeks. We label these transactions using a trace-level classifier applied to \texttt{debug\_traceTransaction} call traces with an Alchemy RPC node. The classifier distinguishes three transaction-level labels: 

\smallskip\noindent\texttt{OFFCHAIN\_DIRECT} transactions show no meaningful pre-swap state read, or use only helper contract or revert-based quoting on the eventual execution route.

\noindent\texttt{ONCHAIN\_EVAL} transactions perform on-chain evaluation before execution, but that evaluation is confined to the route that is eventually executed.

\noindent\texttt{ONCHAIN\_BROAD\_SCAN} transactions evaluate venues beyond the eventual execution set before committing to the winning route.
\smallskip


\parhead{Trace decomposition.}
The classifier decomposes each transaction trace into top-level route-attempt subtrees and identifies the venues used in successful execution. It then compares pre-execution venue reads with this executed-venue set. Reads of venues later used in execution are treated as confined route evaluation, while reads of non-executed venues are treated as evidence of route discovery. Helper contracts,\footnote{This filtering applies to the helper contract address itself. If a helper-call subtree reaches downstream pool contracts, those downstream venues are not filtered and are retained only as weak evidence of \texttt{ONCHAIN\_BROAD\_SCAN}.} token infrastructure, and callback mechanics inside committed swaps are filtered so that execution mechanics are not mistaken for venue search. 

We treat three kinds of trace evidence as informative. First, we identify confirmed direct pool-state reads, such as reserve reads and concentrated-liquidity state queries. Second, we identify balance-based pool reads, where a bot calls \texttt{balanceOf(pool)} on a token contract to infer candidate-pool balances. For singleton-pool-manager designs, we treat \texttt{balanceOf} targets from the successful route attempt as route-local when the executed route uses the singleton manager, since the balance target need not coincide with the executed PoolManager address. Third, we identify revert-based quotes and helper-contract calls used to obtain exact route pricing. We treat these two as weaker evidence than direct or balance-based venue expansion: extra venues reached only through reverted quotes or helper calls do not by themselves trigger \texttt{ONCHAIN\_BROAD\_SCAN}. 
We present the pseudo-code of the classifier implementation in \Cref{app:classifier_implement}.


\parhead{From transaction labels to bot architecture}
\label{sec:bot_week_labeling}
We aggregate transaction-level labels into weekly bot architectures using the implementation patterns introduced in \Cref{sec:preliminaries}. A bot-week is labeled \emph{on-chain discovery} if its sample contains at least one manually validated \texttt{ONCHAIN\_BROAD\_SCAN} transaction. It is labeled \emph{off-chain discovery} if it contains \texttt{OFFCHAIN\_DIRECT} behavior and no credible broad-scan evidence after manual validation. It is labeled \emph{on-chain evaluation} only if all sampled transactions are \texttt{ONCHAIN\_EVAL}. This asymmetry reflects the fact that direct execution and broad venue expansion are clear endpoint behaviors, while confined on-chain evaluation can arise under either off-chain discovery with on-chain re-validation or narrow on-chain discovery.



This rule assumes architecture is persistent within a bot-week, rather than switching transaction by transaction. Broad on-chain scanning requires different execution logic from route-committed execution, so genuine changes to off-chain architecture are more plausibly associated with redeployment or a new bot identity. Because labels are assigned each week, the same bot can still receive different labels across weeks if its observed behavior changes over longer horizons. We then project weekly labels onto the full arbitrage dataset and weekly spam datasets, producing a unified bot-week panel with architecture, successful arbitrage activity, and aggregate spam behavior.

\parhead{Manual validation.}
Since \texttt{OFFCHAIN\_DIRECT} and \texttt{ONCHAIN\_BROAD\_SCAN} transactions should not genuinely coexist within the same bot-week, we manually review mixed cases. The source of false \texttt{ONCHAIN\_BROAD\_SCAN} classification comes from aggregator-mediated execution, especially 1inch Fusion limit order settlement, which can form a closed token loop on-chain without representing DEX-to-DEX price-dislocation arbitrage.\footnote{Consider example tx hash: \href{https://basescan.org/tx/0xe97cd4871c8c0707ef38fe092db088daf035ce464865a5637b65bb63b3ebd7f2}{\texttt{0xe97cd4871c8c0707ef38fe092db088daf035ce464865a5637b65bb63b3ebd7f2}}.} We exclude these transactions from the architecture sample and manually correct the affected bot-week labels.

\parhead{Scope and limitations.}
\added{We infer architecture from successful transactions because failed attempts often do not reveal a realized execution route, making route-confined evaluation and broad scanning indistinguishable. Projecting the weekly label onto spam therefore assumes that successful and failed attempts within a bot-week are generated by the same underlying architecture. While this cannot be easily validated transaction by transaction, the distinct architecture behavior we identified in the later sections provides indirect support for the projection.} 
The classifier is conservative in two further ways: singleton PoolManager designs can hide venue expansion, and sampling up to 50 transactions may miss rare direct-execution or broad-scan transactions. We quantify these effects in \Cref{app:classifier_robustness}. 





\section{Empirical Results}
\label{sec:empirical_result}
We organize the empirical results around the model's comparative statics and Base configuration changes during the sample period \cite{baseconfigchange}. We first establish that the three architecture labels correspond to distinct successful arbitrage and spam profiles. We then study how the prominence of each architecture changes around three empirical episodes: the introduction of Flashblocks, the September on-chain-discovery surge, and the escalation of the minimum-base-fee. Finally, we quantify how each architecture's response to these episodes translates into chain-level outcomes, including priority-fee payments and blockspace consumption.

\subsection{Bot architecture characteristics}
\label{sec:architecture_character}
We begin by documenting how the three bot architectures differ in their observable behavior, both in successful arbitrages and in spam during our sample period. These differences validate that the classifier separates economically distinct architectures rather than spurious trace clusters, and they support the model's key primitives: on-chain discovery bots operate over broader venue sets with lower success rates, and have higher fee exposure.


\parhead{Spam character.}
The three architectures produce strikingly different spam activities. \Cref{tab:bot_label_activity} summarizes full-sample activity by weekly bot label. The table reports the number of bots, successful arbitrages, spam transactions, reverts, probes, average bot-level success rate, spam transaction gas usage, and average cost per arbitrage opportunity. Success rate is computed at the bot level as successful arbitrages divided by total attempts, where total attempts include successful arbitrages and spam transactions. Spam gas is computed as each bot-week's average gas per spam transaction and then summarized across bot-weeks. Average cost per successful arbitrage is total gas cost across successful and failed attempts divided by successful arbitrages.

\begin{table}[t]
\centering
\resizebox{\textwidth}{!}{%
\begin{tabular}{lcccccccc}
\toprule
Bot label & Bots & Arbs & Spam & Reverts & Probes
& \makecell{Success \\ rate} & \makecell{Median / P95 \\ avg. spam tx gas} & \makecell{Avg. cost per\\success}\\
\midrule
Off-chain discovery
& 1,634 & 7.67M & 23.63M & 16.86M & 6.78M 
& 42.51\% & 352K / 1.02M & 0.000095 ETH\\
On-chain evaluation
& 798 & 8.83M & 27.63M & 12.51M & 15.12M
& 18.49\% & 64.8K / 371K & 0.000045 ETH\\
On-chain discovery
& 2,049 & 4.87M & 934.30M & 6.29M & 928.01M
& 4.58\% & 277K / 1.39M & 0.000253 ETH\\
\bottomrule
\end{tabular}
}
\caption{Bot activity by weekly architecture label, June 1, 2025--February 28, 2026. Bot counts are unique addresses observed with each label at least once; categories are not mutually exclusive because labels can vary across weeks.}
\label{tab:bot_label_activity}
\end{table}

Off-chain discovery bots produce the least spam and have the highest average success rate at 42.51\%. The spam they produce is predominantly reverts rather than probes. Because routes are pre-committed, their transactions typically proceed directly to swap execution and revert only when profitability conditions fail close to the end of execution. The resulting reverts traverse part or all of the swap logic, which explains their highest median average spam gas usage (352K) of the three architectures.  Their average cost per successful arbitrage is 0.000095 ETH --- moderate despite heavy per-spam gas usage, because relatively few failed attempts accompany each successful arbitrage.

On-chain evaluation bots produce spam split roughly evenly between reverts and probes, which we interpret as an implementation choice rather than a meaningful architectural distinction: both outcomes terminate a failed attempt after verifying a confined route set. Their evaluation is lightweight and confined to a narrow candidate venue set, so their average spam gas usage is the lowest of the three architectures (median 64.8K, P95 371K). Their average cost per captured opportunity is correspondingly the lowest at 0.000045 ETH, reflecting both cheap spam and a modest success rate.

On-chain discovery bots generate the largest spam activity with 934.30M total spam transactions, almost all of which are probes. This matches their architecture: the bot reads state across many venues to discover whether an opportunity exists at all, and most such reads find none, leaving a probe transaction on-chain. The median average spam gas usage (277K) is higher than that of on-chain evaluation, consistent with broader venue scanning, and the P95 (1.39M) reflects a tail of very wide route searches. Their average cost per opportunity is 0.000253 ETH, much higher than other architectures --- the compound effect of heavy spam transactions and a large number of them per successful arbitrage. The average success rate is correspondingly low at 4.58\%.

\parhead{Successful arbitrage fingerprints.} 
The architectural distinction also shows up in how successful arbitrage transactions are structured. \Cref{fig:gas_used_cdf} plots the cumulative distribution function of successful arbitrage gas usage by bot label. Successful arbitrages by on-chain discovery bots are more gas-intensive, consistent with broad venue scanning making these transactions heavier. We also observe that off-chain discovery transactions use more gas than on-chain evaluation transactions. This does not contradict their route-committed classification: off-chain discovery bots may execute more complex routes, use flash loans,\footnote{Consider example tx hash: \href{https://basescan.org/tx/0x05d5aeccc0ecc11013b62481ea92f47213c200e9ec6bad0c72598d306be40526}{\texttt{0x05d5aeccc0ecc11013b62481ea92f47213c200e9ec6bad0c72598d306be40526}}.} or route through venues such as Balancer, where execution may involve a vault and oracle calls.

\begin{figure}[t]
\centering
\begin{minipage}[t]{0.38\textwidth}
\vspace{0pt}
\centering
\includegraphics[
  width=\linewidth,
  trim=0.25cm 0.25cm 0.25cm 0.25cm,
  clip
]{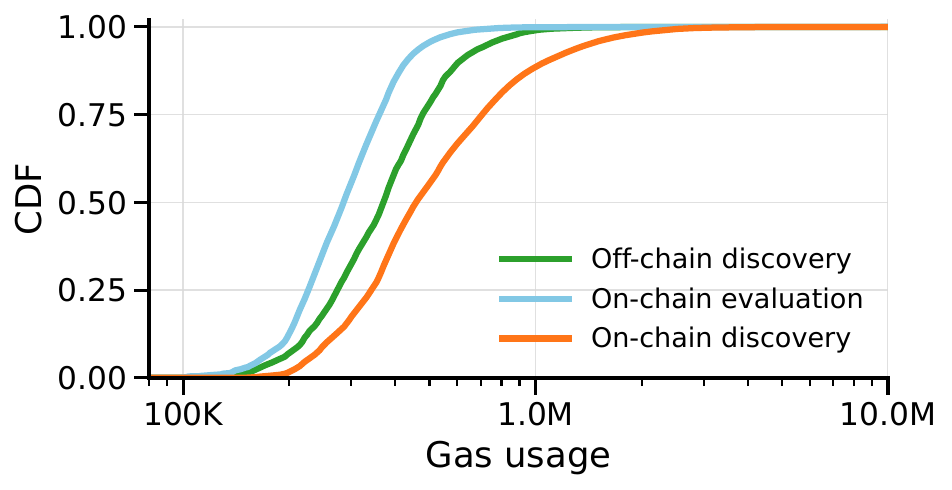}
\caption{Arbitrage gas usage CDF.}
\label{fig:gas_used_cdf}
\end{minipage}
\hfill
\begin{minipage}[t]{0.61\textwidth}
\vspace{0pt}
\centering
\renewcommand{\arraystretch}{1.5}
\resizebox{\linewidth}{!}{%
\begin{tabular}{@{}lcccc@{}}
\toprule
& \multicolumn{2}{c}{Bot-level architecture}
& \multicolumn{2}{c}{Transaction-level footprint} \\
\cmidrule(lr){2-3}\cmidrule(l){4-5}
Bot label
& \makecell{Median \\ calldata length}
& \makecell{Average / Median \\venue commitment}
& \makecell{Median \\ pre-reads}
& \makecell{oracle usage} \\
\midrule
Off-chain disc. & 516 & 0.607 / 0.971 & 0  & 22.61\%\\
On-chain eval. & 150 & 0.535 / 0.886 & 3  & 4.61\%\\
On-chain disc. & 91  & 0.184 / 0.000 & 40 & 0.08\%\\
\bottomrule
\end{tabular}
}
\captionof{table}{Sampled arbitrage trace-level metrics.}
\label{tab:architecture_metrics}
\end{minipage}
\end{figure}

\Cref{tab:architecture_metrics} reports complementary architecture and footprint measures from the sampled successful arbitrage transactions. Calldata length and venue commitment are summarized at the bot level to avoid overweighting high-activity bots; pre-execution reads and oracle usage are summarized at the transaction level because they measure per-transaction footprint. Venue commitment is a conservative address-level proxy for route pre-commitment, defined as the fraction of executed venue addresses embedded in calldata.

Off-chain discovery bots have the longest median calldata and the highest venue commitment, and their transactions have a median of zero pre-execution reads, consistent with pre-committed routes carried in the transaction itself and direct execution behavior. On-chain discovery bots show the opposite pattern with the shortest median calldata, zero median venue commitment, and 40 median pre-execution reads in their transactions, indicating that the route is resolved during execution through dense state reads across many venues. On-chain evaluation bots sit between the two, with modest pre-reads against the confined set of venues they subsequently execute on.

The bot-level average venue-commitment score for off-chain discovery is 0.607, lower than a literal route-committed interpretation might suggest. This reflects a limitation of calldata-based address matching rather than necessarily incomplete route commitment. Venue addresses may be hardcoded in the bot contract, derived through helper calls, or represented indirectly through singleton AMM interfaces, so the metric can record low or zero commitment even when the route is pre-specified.\footnote{For example, transaction \href{https://basescan.org/tx/0x009f0869c035769e1eee4ff1651e1949ec6f2054cae7866c916a417a0d407d17}{\texttt{0x009f0869c035769e1eee4ff1651e1949ec6f2054cae7866c916a417a0d407d17}} executes an ETH--USDT--USDC--ETH arbitrage across three Uniswap v4 pools. The calldata contains PoolKey-like parameters, such as tokens and hook addresses, but not the PoolManager address, contributing to the venue commitment metric being 0 for this transaction.} We therefore interpret venue commitment as a conservative proxy for route pre-commitment. The combination of long calldata and zero median pre-execution reads remains the stronger fingerprint of off-chain route commitment.


\parhead{Time-series overview}

We next document the time-series evolution of the prominence of each architecture during the observed period. \Cref{fig:arb_spam_timeseries} presents the weekly successful arbitrage, USD volume, and spam activity by architecture. Across the full sample, on-chain discovery bots account for 22.8\% of the successful arbitrages and 28\% of volume, but produce nearly 95\% of the total spam activity. \Cref{fig:bot_count_arb_share_timeseries}A plots the weekly count of active bots by architecture. In most weeks, 30--50 active on-chain discovery bots generate roughly 20M spam transactions per week. By contrast, off-chain discovery and on-chain evaluation have larger active populations, but generate substantially less spam while accounting for a larger share of successful arbitrages.
\begin{figure}[t]
    \centering
    \includegraphics[width=1\linewidth]{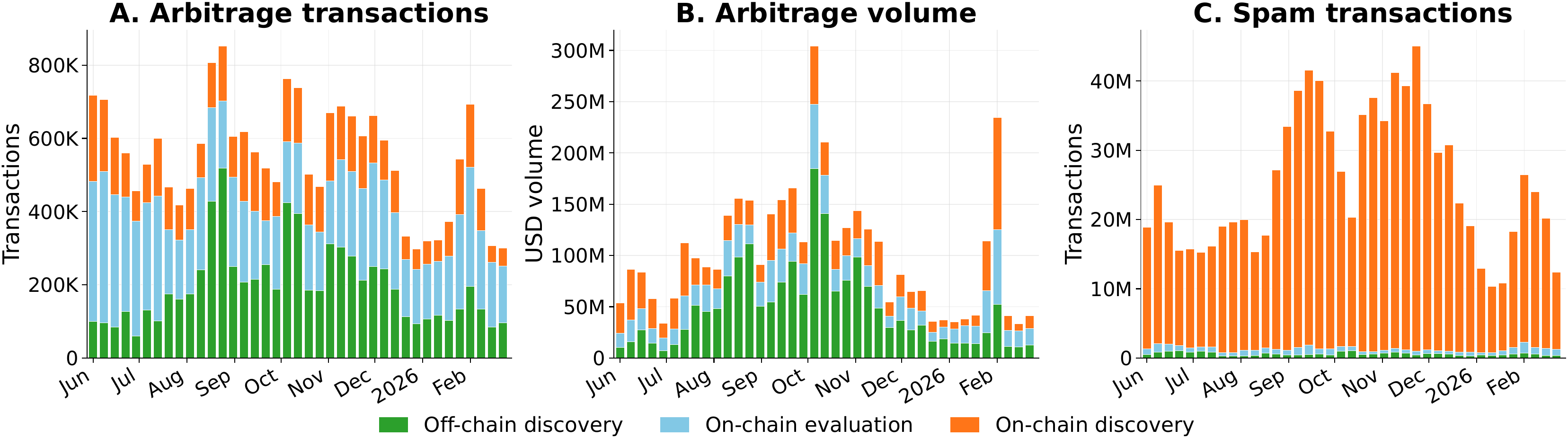}
    \caption{Weekly successful arbitrage transaction count, USD volume, and spam transaction count by each bot architecture label.}
    \label{fig:arb_spam_timeseries}
\end{figure}
\Cref{fig:bot_count_arb_share_timeseries}B--C tracks active bots and volume shares by architecture. The series highlights three episodes we study below: Flashblocks on July 7, 2025; a temporary on-chain discovery surge in early September 2025; and the minimum-base-fee escalation beginning December 4, 2025. These episodes correspond to the model's three comparative statics: finer ordering, opportunity-access shocks, and fee-floor increases.

\begin{figure}[t]
    \centering
    \includegraphics[width=1\linewidth]{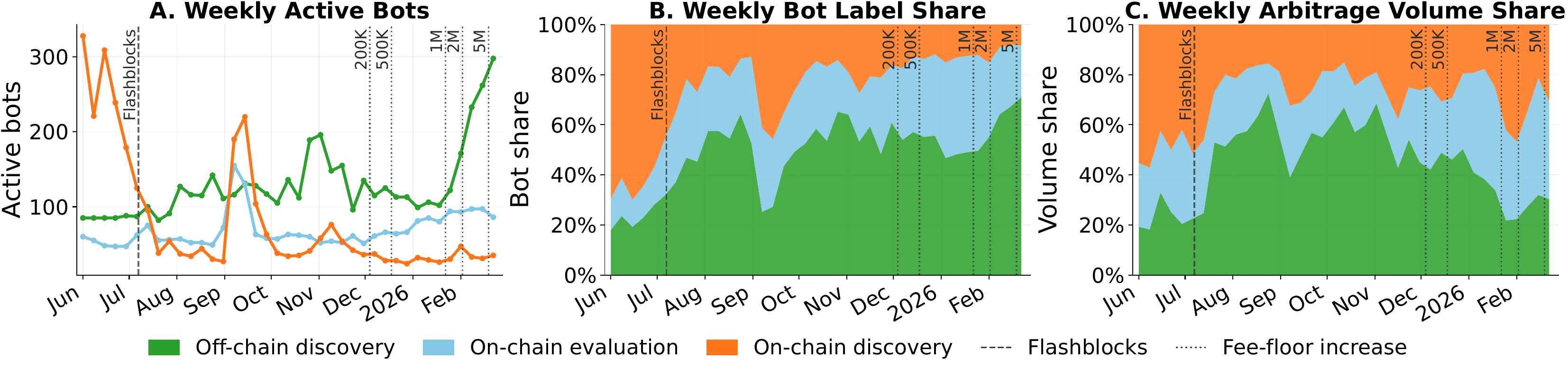}
    \caption{Weekly number of active bots per label, weekly active bot label share, and weekly volume share per bot label. Vertical lines mark configuration changes of Flashblocks and fee floors.}
    \label{fig:bot_count_arb_share_timeseries}
\end{figure}


\subsection{Flashblocks}
\label{sec:flashblocks}
The Flashblocks activation on July 7, 2025, introduced 200ms sub-block pre-confirmations within Base's 2-second block cycle, and each Flashblock has an incrementally increasing space capacity \cite{flashblocksconfig}, substantially finer than the full-block ordering. Our model predicts two patterns around this transition: a decline in on-chain discovery participation (cf. \Cref{thm:fb_competitiveness}); selection against more complex on-chain discovery strategies (cf. \Cref{prop:selection_on_complexity}) and towards simpler on-chain discovery survivors (cf. \Cref{cor:simpler_survivors}). We empirically validate these predictions using a four-week symmetric event window around the rollout: pre weeks June 8 – July 6, 2025, post weeks July 13 – August 10, 2025, with the event week itself excluded.

\parhead{Population and volume shift.}
The composition of active bots changes sharply after the Flashblocks introduction. The number of total distinct active bots declines by 53\%, falling from 983 in the pre-window to 462 in the post-window. The contraction is concentrated almost entirely in on-chain discovery architecture: its share of active bots falls from 63.1\% to 25.1\%, while off-chain discovery rises from 23.4\% to 46.7\% and on-chain evaluation rises from 13.4\% to 28.2\%. As illustrated in \Cref{fig:bot_count_arb_share_timeseries}A, the on-chain discovery population peaks above 300 weekly active in mid-June, declines through the pre-window, and stabilizes around 30 weekly active bots in the post-window. The off-chain discovery population, in contrast, expands steadily from roughly 85 to over 150 weekly active bots by early August.

\Cref{fig:bot_count_arb_share_timeseries}C presents a consistent story in terms of arbitrage volume. On-chain discovery's share of weekly arbitrage USD volume falls from roughly 50-60\% pre-Flashblocks to around 30\% post-Flashblocks and remains at that level for the remainder of the observation window. Off-chain discovery expands from roughly 20\% to over 50\% of volume, and on-chain evaluation remains broadly stable. These movements match \Cref{thm:fb_competitiveness}: finer ordering reduces participation from on-chain discovery architecture and increases off-chain participation.

A natural concern is that on-chain discovery was already declining before Flashblocks, falling from 309 weekly active bots by June 15 to 125 by July 6. This decline likely reflects the unwinding of a May-June opportunity cycle around AI-agent tokens on Base.\footnote{Examples include KTA's May price run-up, RDAC's May token launch and Binance listing, VIRTUAL's late-May peak and June correction, and AIXBT's late-May Binance.US listing. We use these events only as context for a contemporaneous opportunity cycle; the empirical results below rely on the on-chain data rather than on attributing the decline to any specific token.} However, a purely demand-side explanation is less plausible: on-chain discovery's volume share remains stable during the pre-window, and around Flashblocks, aggregate arbitrage activity rebounds (cf. \Cref{fig:arb_spam_timeseries}), yet on-chain discovery bots continue to contract from 179 in the week before Flashblocks to 38 three weeks after.

\parhead{Selection against on-chain discovery.}

The composition shift around Flashblocks is driven primarily by selection rather than within-bot label switching: among the 174 bots active in both pre- and post-windows, 167 retain the same label, so the shift mainly reflects differential exit and replacement. The selection margin is sharpest within on-chain discovery. To characterize it, we compare three cohorts within the on-chain discovery architecture: bots active only in the pre-window (\emph{selected out}), only in the post-window (\emph{entrants}), and in both windows under the same label (\emph{survivors}). \Cref{tab:flashblocks_onchain_cohort} reports the bot-level architecture and transaction trace-level footprint metrics calculated from sampled arbitrages for three cohorts, using only their bot-weeks labeled as on-chain discovery.

\begin{table}[t]
\centering
\resizebox{\linewidth}{!}{%
\begin{tabular}{@{}lrrrrrrr@{}}
\toprule
& & \multicolumn{3}{c}{Bot-level architecture}
& \multicolumn{3}{c}{Transaction-level footprint} \\
\cmidrule(lr){3-5}\cmidrule(l){6-8}
Cohort & Bots
& \makecell{Median \\ calldata length}
& \makecell{Average venue \\ commitment}
& \makecell{Median venue \\ commitment}
& \makecell{Median \\ pre-reads}
& \makecell{Median \\ arbitrage tx gas}
& \makecell{Median of avg.\\ spam tx gas} \\
\midrule
Selected out (pre only) & 741 & 36 & 0.044 & 0.000 & 57 & 855K & 671K \\
Entrants (post only) & 120 & 244 & 0.258 & 0.000 & 46 & 630K & 272K \\
Survivors (pre $\rightarrow$ post) & 37
& 163 $\rightarrow$ 302
& 0.345 $\rightarrow$ 0.340
& 0.000 $\rightarrow$ 0.000
& 39 $\rightarrow$ 12
& 662K $\rightarrow$ 437K
& 226K $\rightarrow$ 103K \\
\bottomrule
\end{tabular}
}
\caption{On-chain discovery cohort composition around Flashblocks. Bot-level metrics first summarize each bot within the relevant window and then summarize across bots. Transaction-level metrics summarize sampled arbitrage transactions and spam bot-week medians within the same four-week pre- and post-event windows.}
\label{tab:flashblocks_onchain_cohort}
\end{table}


The 741 selected-out on-chain discovery bots have the highest scan footprint in the dataset. Their median bot-level calldata length is only 36 bytes, median venue commitment is zero, and their transaction-level footprint is the heaviest, with 57 median pre-reads and the highest gas usage for both arbitrage and spam transactions. In other words, the bots most heavily dependent on short-calldata, low-commitment, high-read scanner logic are precisely the ones that vanish after Flashblocks.

Relative to the selected-out bots, the entrant cohort is markedly less scan-intensive --- longer calldata, higher (average) venue commitment, fewer pre-reads, and lower gas usage on successful arbitrages and spam transactions. These patterns likely help explain why they can enter after Flashblocks.
The survivors are simpler still. Even in the pre-window, surviving on-chain discovery bots exhibit the least search behavior across the three cohorts, suggesting that survival is concentrated among bots whose on-chain search is already relatively lean. 
By contrast, the corresponding architecture changes for off-chain discovery and on-chain evaluation population are modest around Flashblocks (cf. \Cref{appendix:flashblocks_result}, \Cref{tab:flashblocks_offchain_architecture}). These results confirm \Cref{prop:selection_on_complexity,cor:simpler_survivors} that finer ordering selects against on-chain discovery with higher per-attempt scan-footprint and towards less scan-intensive ones.

\parhead{Adaptation within on-chain discovery.}

Same-label survivor behavior and a small number of label switches show that adaptation exists, but it is quantitatively smaller than selection. In \Cref{tab:flashblocks_onchain_cohort}, surviving on-chain discovery bots become less scan-intensive after Flashblocks: their calldata becomes longer, venue commitment rises, and median pre-execution reads fall. This pattern suggests that survivors were not only leaner before the shock, but also adjusted further toward more committed and less scan-intensive execution after Flashblocks. In other words, adaptation operates on the intensive margin among the bots that remain competitive, while the extensive-margin exit of high-search bots remains the dominant response.

We observe a similar narrowing among the small set of bots that switch from on-chain discovery to on-chain evaluation. 
Four bots switch from consistently on-chain discovery while active before Flashblocks to consistently on-chain evaluation while active afterward, consistent with a shift away from broad venue scanning toward evaluation confined to the eventual execution route. One example is \texttt{0x57b5f8d2b95cd55086bf0c954051547271b68fa9}: in the pre-window, 97 of its 100 sampled arbitrage transactions are labeled \texttt{ONCHAIN\_BROAD\_SCAN}, while in the post-window its sampled transactions no longer exhibit broad-scan behavior, although the post-window sample contains only 7 transactions. We therefore interpret these switches cautiously. They need not imply that the bot's discovery process moved fully off-chain; rather, they indicate that the observable on-chain component of search became narrower, moving from broad scanning toward route-confined evaluation. Together with the survivor metrics, this evidence supports a secondary adaptation margin, but one that is much smaller than the selection effect documented above.

The same selection and adaptation patterns hold in a two-week window around Flashblocks, suggesting that the result is not driven by the choice of event window or by later post-event weeks (cf. \Cref{appendix:flashblocks_result}).

\parhead{Attempt efficiency and spam concentration.}

The selection documented above concerns the transaction-level footprint of arbitrage attempts. It does not imply that aggregate spam activity should fall mechanically, because spam also depends on how many failed attempts each bot sends. \Cref{tab:flashblocks_attempt_efficiency_spam} compares spam production and attempt efficiency by architecture over the four-week pre- and post-Flashblocks windows.


The spam response is highly asymmetric. Off-chain discovery and on-chain evaluation both reduce total spam and spam per active bot-week. On-chain discovery, however, remains spam-intensive: despite active bots falling from 782 to 157, total spam rises from 68.43M to 70.50M, and spam per active bot-week rises from 72.49K to 314.71K. Because on-chain discovery dominates spam in both windows, aggregate spam is nearly flat.

\begin{table}[t]
\centering
\resizebox{\linewidth}{!}{%
\begin{tabular}{lrrrrrr}
\toprule
Bot label
& \makecell{Success rate \\ pre $\rightarrow$ post}
& \makecell{Gas cost / success \\ pre $\rightarrow$ post}
& \makecell{Total spam txs \\ pre $\rightarrow$ post}

& \makecell{Spam per\\active bot-week}
& \makecell{Active bots \\ pre $\rightarrow$ post} \\
\midrule
Off-chain disc.
& 35.8\% $\rightarrow$ 43.6\%
& 0.000140 $\rightarrow$ 0.000113
& 3.73M $\rightarrow$ 1.78M
(-52.3\%)
& 11.13K $\rightarrow$ 4.83K
& 134 $\rightarrow$ 207 \\
On-chain eval.
& 26.1\% $\rightarrow$ 23.1\%
& 0.000015 $\rightarrow$ 0.000027
& 3.66M $\rightarrow$ 2.54M
(-30.7\%)
& 18.70K $\rightarrow$ 10.45K
& 86 $\rightarrow$ 113 \\
On-chain disc.
& 5.7\% $\rightarrow$ 3.8\%
& 0.000293 $\rightarrow$ 0.000350
& 68.43M $\rightarrow$ 70.50M
(+3.0\%)
& 72.49K $\rightarrow$ 314.71K
& 782 $\rightarrow$ 157 \\
\bottomrule
\end{tabular}
}
\caption{Attempt efficiency and spam production around Flashblocks. 
Success rates and gas cost per successful arbitrage are bot-level averages within each window. Spam per active bot-week normalizes spam count by the number of labeled bot-week observations. Bot counts are not mutually exclusive because labels may vary across weeks.}
\label{tab:flashblocks_attempt_efficiency_spam}
\end{table}


Attempt efficiency moves in a similar direction. Flashblocks improves attempt efficiency for off-chain discovery: average success rate rises from 35.8\% to 43.6\%, and gas cost per successful arbitrage falls from 0.000140 to 0.000113 ETH. This is consistent with reduced stale-state risk between off-chain opportunity discovery and execution (cf. \Cref{rem:offchain_reliability}). On-chain discovery moves in the opposite direction, with lower success rates and higher gas cost per success. Despite fewer failed attempts, the attempt efficiency of on-chain evaluation declines, as their successful arbitrages fall (1.403M pre-window $\rightarrow$ 0.855M post-window).

The persistence of on-chain discovery spam reflects a response in attempt intensity among the bots that remain. Flashblocks selects against many broad scanner bots and leaves less scan-intensive survivors, but it does not directly price attempt frequency. By shortening the window in which each opportunity remains available for on-chain discovery, finer ordering can lower per-attempt yield; surviving discovery bots can respond by probing more frequently. In this sense, Flashblocks changes who participates in on-chain discovery and narrows the form of surviving search, but does not by itself eliminate high-frequency probing.


\parhead{Summary.}
The Flashblocks event study supports the model's central mechanism: finer ordering reduces the competitiveness, and therefore the participation, of on-chain discovery, especially among the most scan-intensive bots. The dominant margin is selection: broad scanners exit, same-label survivors are already narrower and simplify further, and the market shifts toward route-committed off-chain discovery. Off-chain discovery becomes more attempt-efficient after Flashblocks, consistent with the stale-state channel: shorter ordering intervals make route-committed execution more reliable. However, this architectural selection does not translate into a proportional reduction in total spam. Aggregate spam remains nearly flat because on-chain discovery spam becomes concentrated among fewer, higher attempt-intensity entrants and survivors. Thus, Flashblocks reduces broad on-chain discovery and improves off-chain execution efficiency, but its effect on aggregate spam is partially offset by higher attempt intensity among the remaining on-chain discovery bots.

\subsection{The September surge}
\label{sec:sep_surge}
The post-Flashblocks on-chain discovery population stabilizes at roughly 30-50 weekly active bots through August 2025. However, this number spikes sharply in early September: active on-chain discovery bots rise to 220 in the week beginning September 14, before quickly falling back to 62 by September 28 and stabilizing around 35 early October (cf. \Cref{fig:bot_count_arb_share_timeseries}A). A smaller spike also appears among on-chain evaluation, while the off-chain discovery remains comparatively stable. 
Because the aggregate time series does not show a broad, sustained market-wide expansion (cf. \Cref{fig:arb_spam_timeseries}), we interpret the episode as a localized entry wave around newly attractive venues rather than a broad expansion in arbitrage opportunities.
For the remainder of this section, we refer to the four weekly buckets from September 7--13 through September 28--October 4 as the September surge window.



\parhead{A short-lived entrant wave.}
The entrant composition shows that the surge is bimodal: a few serious entrants and a much larger mass of short-lived triers. Across the September surge window, 504 unique on-chain discovery bots are active, of which 
464 (92.1\%) are both first and last seen within these weeks. Most are active only briefly: 426 bots appear only in one week, and 470 appear in at most two weeks. The same pattern appears in transaction intensity: median transient on-chain discovery bots execute only 29 surge-window arbitrages, compared with 262.5 for non-transient bots (cf. \Cref{appendix:sep_surge_result}, \Cref{tab:september_entrant_wave}).

\parhead{AVNT/MIRROR opportunity wave.}
The entrant wave is centered on two newly launched tokens: AVNT, launched on September 9 \cite{avnt}, and MIRROR, launched on September 8 \cite{mirror}. \Cref{fig:avnt_mirror}A shows that AVNT and MIRROR form a visible share of arbitrage volume from September 7 to September 28, and account for much of the contemporaneous increase in total volume.
After September 28, AVNT/MIRROR arbitrage volume falls quickly, coinciding with the disappearance of the temporary on-chain discovery entrant population.

\begin{figure}[t]
    \centering
    \includegraphics[width=1\linewidth]{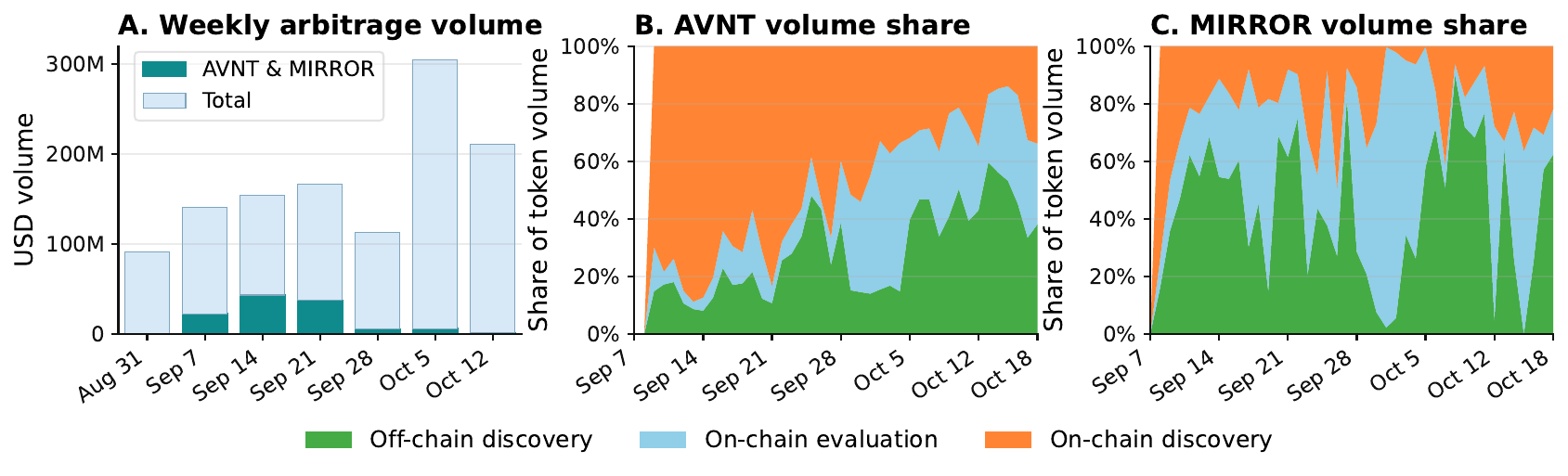}
    \caption{AVNT and MIRROR arbitrage activity around the surge window. Panel A shows weekly AVNT/MIRROR volume relative to total volume, with the x-axis showing the week-start date. Panels B and C show AVNT and MIRROR daily volume share, respectively, by bot architecture.}
    \label{fig:avnt_mirror}
\end{figure}

AVNT and MIRROR account for a large share of arbitrage activity during the surge window. From September 7 to October 4, AVNT or MIRROR appears in 420,541 arbitrages (19.3\% of all arbitrages) and \$109.1M in volume (19.0\%). The concentration is much stronger within on-chain discovery: 42.8\% of arbitrages and 50.2\% of volume involve AVNT/MIRROR, spanning 412 of the 504 active on-chain discovery bots. Most of this activity is driven by AVNT (\$79.1M vs.\ \$0.5M for MIRROR).

The concentration is especially strong among short-lived triers. Of the 464 transient on-chain discovery bots, 84.3\% execute at least one AVNT/MIRROR arbitrage, and 54.7\% of them have a majority (>50\%) of their activity in these tokens during the surge window.
By contrast, persistent on-chain discovery bots are more diversified: median AVNT/MIRROR share is only 2.0\%, and 35 of 40 have a majority of their activity outside these tokens. 

The swarm does not persist. \Cref{fig:avnt_mirror}B--C show that on-chain discovery's share of AVNT/MIRROR volume falls quickly post-surge, while the other two architectures absorb the bulk of the flow. The opportunity itself does not disappear --- AVNT continues to generate arbitrage activity post-surge (Oct 5, 2025 to Feb 28, 2026), with 120,872 transactions across 308 bots, compared with 396,957 transactions across 619 bots during the surge weeks. What disappears is the surge-window on-chain discovery cohort: only 13 of the 504 September on-chain discovery bots remain active on post-surge AVNT flow, and all on-chain discovery bots account for 29.5\% of post-surge AVNT volume, while off-chain discovery bots capture the largest share of post-surge AVNT volume (41.4\%). The MIRROR pattern is similar, with only the 7 surge-window cohort bots surviving post-surge. Off-chain discovery bots capture the largest 60.9\% post-surge MIRROR volume, while on-chain discovery bots fall to 24\%.

We interpret this transition as evidence that on-chain discovery is temporarily more competitive when new token routes are not yet routinized. This need not mean that existing on-chain discovery contracts discover new pools faster. Instead, early token-launch markets can generate dispersed retail flow, uncertain venue quality, and many short-lived route candidates, making probing attractive. Bots may pass venue parameters through calldata, use generic scanning logic, or deploy cloned contracts specialized to new route families \cite{ozan_lioba_paper}. As the routes become easier to monitor and parameterize off-chain, activity shifts toward on-chain evaluation and off-chain discovery, while most short-lived on-chain discovery entrants exit.  



\parhead{Spam burden.}
The entrant wave also generates a temporary spam shock: on-chain discovery bots produce 147M spam transactions during the surge window (cf. \Cref{fig:arb_spam_timeseries}C). Most of the level comes from short-lived entrants. The 464 transient bots generate 105.7M spam transactions, or 72\% of on-chain discovery spam. The remaining 40 non-transient bots generate 41.3M spam transactions, but are far more intensive on a per-bot basis and account for a larger arbitrage activity per bot (403,464 arbitrages across 40 bots, versus 185,997 across 464 transient bots). Thus, the surge combines a mass-entry spam shock from transient bots with a smaller persistent core of high-intensity scanners.

\parhead{Summary.}
The September surge provides an example of the model's opportunity-shock channel (cf. \Cref{prop:launch}). On-chain discovery briefly revives around two newly launched token families in an otherwise stable post-Flashblocks regime, then recedes as launch-specific activity fades and the remaining flow is increasingly handled by other architectures. The associated spam surge comes mostly from short-lived triers that exit afterward.
Since the episode overlaps the September 17 per-transaction gas-limit change \cite{baseconfigchange}, we do not interpret it as a clean event study of that policy. A later November rebound contrasts with September: volatility in already active token families mainly expands off-chain discovery, while a smaller on-chain discovery response still generates a visible spam increase (cf. \Cref{appendix:sep_surge_result}).

\subsection{Minimum base fee escalation}
\label{sec:fee_escalation}

Base introduced a minimum base fee on December 4, 2025 and increased it in five steps: 200K WEI on December 4, 500K WEI on December 18, 1M WEI on January 22, 2M WEI on February 2, and 5M WEI on February 19. In the model, the fee floor operates through the cost channel: architectures with greater fee exposure per captured opportunity become less viable as the fee floor rises. We expect the fee ramp to disproportionately affect on-chain discovery, where repeated probing and low success rates make the value-cost ratio especially fee-sensitive (cf. \Cref{prop:fee_floor,prop:fee_exposure_selection,rem:fee_exposure_decomposition}). Because the fee changes are tightly spaced, we analyze the overall weekly population shift and use one-week local comparisons around each step to study selection on value-cost economics.

\parhead{Population shift.}
The fee ramp coincides with a gradual reallocation toward off-chain discovery. \Cref{fig:bot_count_arb_share_timeseries}B shows on-chain discovery's active-bot share falling steadily across the ramp, while off-chain discovery rises, especially from late January 2026 onward. By the final sample week, on-chain discovery accounts for 8.4\% of active bots and off-chain discovery for 71.1\%, with on-chain evaluation comparatively stable. The shift is mainly selection and replacement rather than label switching: only three surviving bots switch labels during the fee escalation regime. The remaining on-chain discovery bots still capture non-trivial arbitrage flow, indicating that a smaller high-productivity core remains viable.

\begin{figure}[t]
\centering
\includegraphics[width=\linewidth]{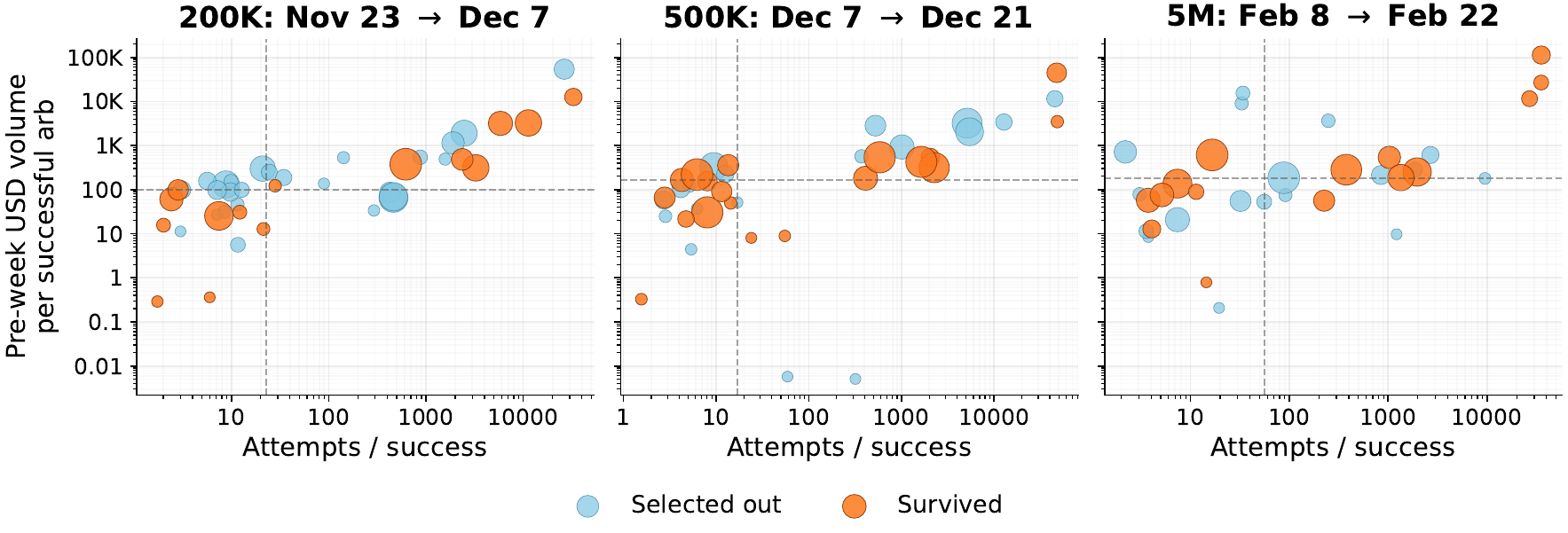}
\caption{On-chain discovery bots active in the pre-step window for each fee increase, plotted by attempts per successful arbitrage and USD volume per successful arbitrage on log scales. Color indicates whether the bot remains active in the post-step window (orange) or is absent (blue). Marker size scales with total volume. The vertical and horizontal dashed lines mark cohort medians. Leftward indicates lower attempt intensity; upward indicates higher value per success, so the upper-left region corresponds to stronger value-cost margins.}
\label{fig:fee_step_scatter}
\end{figure}

\parhead{Selection on value-cost margin.}
\Cref{rem:fee_exposure_decomposition} predicts that fee floors raise the hurdle rate for on-chain discovery participation: bots whose opportunity access does not justify repeated attempts at higher cost should exit, regardless of their attempt intensity in isolation. \Cref{fig:fee_step_scatter} plots pre-step on-chain discovery bots in attempts per success and USD volume per success at the 200K, 500K, and 5M fee steps, with marker size proportional to each bot's total pre-step arbitrage volume. Volume per success proxies for the value side of the value-cost ratio: large-size successful arbitrages imply more gross opportunity value available to cover failed-attempt costs. We discuss the 1M and 2M steps separately because they overlap a mid-ramp opportunity shock (cf. \Cref{appendix:fee_result}).

The selection pattern is two-dimensional: neither attempt intensity nor volume per success alone predicts survival, with both surviving and selected-out bots present across most regions of each panel. Bots with high volume per success and low attempts per success tend to survive: they generate value without incurring high failed-attempt costs. Bots with low volume per success and high attempts per success are disproportionately selected out because their value-cost ratio is unfavorable, and the fee floor pushes them below viability. Bots with both high attempts and high volume per success survive only when their successful opportunities are valuable enough to offset the cumulative cost of repeated attempts; in the 5M panel, this region is anchored by a small cluster of markers operating at \$10,000--100,000 per successful arbitrage. Thus, the fee floor does not select against attempt intensity alone; it selects against attempt intensity unsupported by sufficient opportunity value.

A cohort-level value-per-attempt comparison before and after each fee step, reported in \Cref{appendix:fee_result}, \Cref{tab:fee_step_value_per_attempt}, supports this interpretation. At the clean 200K and 5M steps, surviving bots' aggregate USD volume per attempt rises from 0.36 to 0.87 and from 0.46 to 1.31, respectively, consistent with pruning low-yield attempts or shifting toward higher-value opportunities. The 500K step moves in the opposite direction, but this window coincides with a broad market contraction in late December (cf. \Cref{fig:arb_spam_timeseries}). 

\parhead{Stability in per-attempt footprint.}
\Cref{tab:fee_ramp_endpoint_trace_features} compares bot-level and transaction-level metrics for on-chain discovery bots at the pre-ramp endpoint, the week beginning November 24, 2025, and the post-ramp endpoint, the week beginning February 23, 2026. Unlike Flashblocks, the fee ramp does not appear to select for less scan-intensive on-chain discovery attempts at the trace-level. Therefore, broad scanning can still be profitable when opportunities are valuable enough, but low-yield probing becomes harder to sustain.
\begin{table}[t]
\centering
\resizebox{\linewidth}{!}{%
\begin{tabular}{@{}lrrrrrrr@{}}
\toprule
& & \multicolumn{3}{c}{Bot-level architecture}
& \multicolumn{3}{c}{Transaction-level footprint} \\
\cmidrule(lr){3-5}\cmidrule(l){6-8}
Period & Bots
& \makecell{Median \\ calldata length}
& \makecell{Average venue \\ commitment}
& \makecell{Median venue \\ commitment}
& \makecell{Median \\ pre-reads}
& \makecell{Median \\ arbitrage tx gas}
& \makecell{Median avg.\\ spam tx gas} \\
\midrule
Pre-ramp endpoint & 42 & 371 & 0.288 & 0.347 & 14 & 471K & 73K \\
Post-ramp endpoint & 35 & 304 & 0.278 & 0.321 & 21 & 497K & 72K \\
\bottomrule
\end{tabular}
}
\caption{Endpoint on-chain discovery architecture around the minimum-base-fee ramp. Bot-level architecture metrics first summarize each bot within the endpoint week and then summarize across bots. Transaction-level footprint metrics summarize sampled arbitrage transactions and spam bot-week medians in the same endpoint weeks.}
\label{tab:fee_ramp_endpoint_trace_features}
\end{table}


\parhead{Summary.}
These results support the model's cost-channel in \Cref{prop:fee_floor,prop:fee_exposure_selection}. As fee floors raise per-attempt costs, on-chain discovery, the highest-exposure architecture, contracts from 22.0\% to 8.4\% of active bots. The selection is against repeated attempts unsupported by sufficient opportunity value, while broad scanning remains viable when opportunity value is high. Thus, higher fee floors shift the on-chain discovery viability frontier along the value-cost margin, distinct from Flashblocks' blockspace channel, which selects against heavy venue scanning in early ordering slices.

\subsection{Protocol revenue and blockspace efficiency}
\label{sec:chain_efficiency}

The previous sections documented how bot architectures responded to Base's configuration changes. We now ask whether these responses improved chain-level outcomes, including protocol fee revenue composition and blockspace consumption.\footnote{We note that the base fees on Base accrue to the sequencer rather than being burned as on Ethereum. Our analysis isolates the priority fee component of broader sequencer revenue.}
These outcomes speak to both protocol revenue quality and the blockspace externality caused by spam.


\parhead{Priority fees by architecture.}
Across the full sample, off-chain discovery successful arbitrages pay the highest median priority fee per transaction, at \SI{8.66e-7} ETH, compared with \SI{4.72e-7} ETH for on-chain evaluation and \SI{5.02e-7} ETH for on-chain discovery. This ordering is consistent with the economic role of off-chain route commitment. Because these bots can simulate routes and optimize input sizes before submitting, they are better positioned to identify higher-value opportunities and bid more aggressively for prioritized inclusion. Thus, the higher priority fees likely reflect both better opportunities and competition for more valuable arbitrages.  

\begin{figure}[t]
    \centering
    \includegraphics[width=\linewidth]{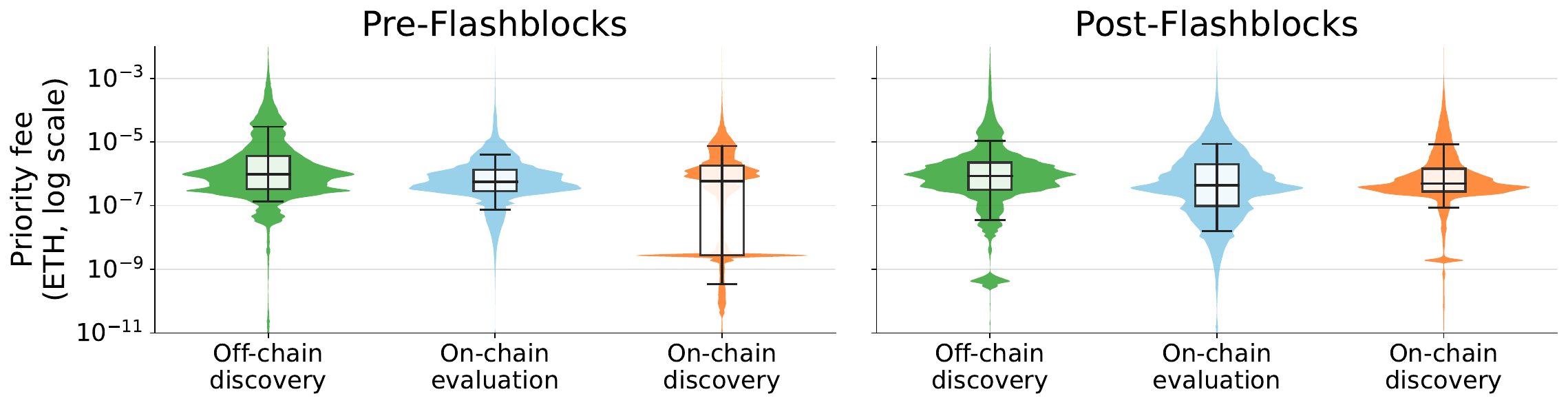}
    \caption{Priority-fee distributions for successful arbitrages by architecture before and after Flashblocks (3.04M arbitrages pre, 17.40M arbitrages post).}
    \label{fig:prio_fee_arb_distribution}
\end{figure}

\Cref{fig:prio_fee_arb_distribution} plots priority-fee distributions for successful arbitrages before and after Flashblocks by architecture. We note that the pre-Flashblocks window contains 5 weeks, while the post-Flashblocks window contains 33 weeks. The largest distributional change occurs for on-chain discovery. Before Flashblocks, successful on-chain discovery arbitrages exhibit a pronounced low-priority-fee mass, which we interpret as low-willingness-to-pay probabilistic search. After Flashblocks, this low-fee mass shrinks, and the on-chain discovery distribution moves closer to the other architectures. This pattern is consistent with Flashblocks selecting out the broadest low-revenue scanners and leaving a smaller population whose attempts are more targeted (cf. \Cref{tab:flashblocks_onchain_cohort}).
The same qualitative pattern appears on the spam side, although spam is observed only at the bot-week aggregate level (cf. \Cref{appendix:protocol_result}, \Cref{fig:prio_fee_spam_dist}). 

Taken together, the priority-fee evidence suggests that Flashblocks changes not only which bots remain active, but also how surviving on-chain discovery bots bid for inclusion. The low-priority-fee mass of on-chain discovery shrinks after Flashblocks, and the remaining on-chain discovery transactions pay priority fees closer to the other architectures. Thus, the post-Flashblocks on-chain discovery population is smaller but bids more aggressively.

A similar, though less clean, pattern appears across the minimum-base-fee ramp. Median priority fees for successful on-chain discovery arbitrages rise from \SI{4.58e-7} ETH in the pre-ramp endpoint week to \SI{2.69e-6} ETH in the post-ramp endpoint week, consistent with low-value (and therefore low-fee) attempts becoming less viable after repeated attempts become more costly (cf. \Cref{appendix:protocol_result}, \Cref{fig:prio_fee_week26_vs_week39}). Because we observe only one post-ramp week, we interpret it cautiously.


\parhead{Protocol revenue composition.}
We next aggregate priority fees across arbitrage attempts to examine protocol revenue. \Cref{fig:priority_fee_composition} decomposes cumulative and weekly priority fees by architecture and attempt outcome. 
\begin{figure}[t]
    \centering
    \includegraphics[width=1\linewidth]{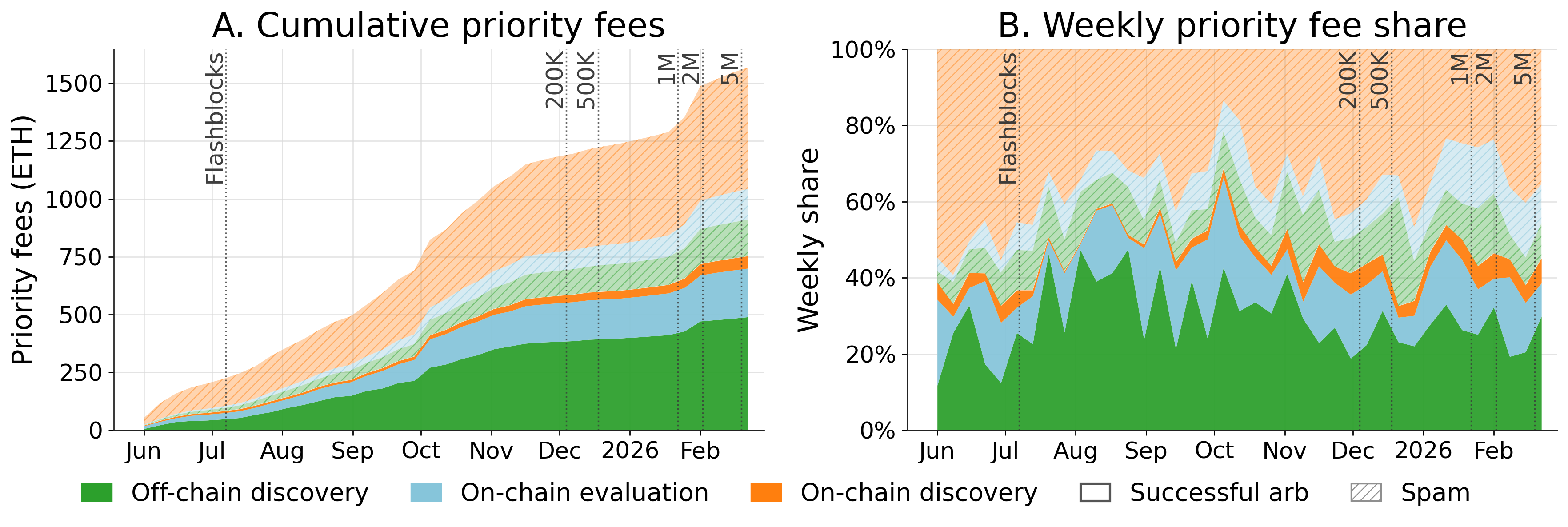}
    \caption{Priority-fee payments by bot architecture and transaction outcome. Panel A reports cumulative priority fees from arbitrage attempts. Panel B reports weekly priority-fee shares. Colors denote bot architecture; solid areas denote successful arbitrages and hatched areas denote spam.}
    \label{fig:priority_fee_composition}
\end{figure}
The level series in \Cref{fig:priority_fee_composition}A is highly opportunity-driven: protocol revenue accumulates faster during high-volume periods, especially around the October and late-January episodes (cf. \Cref{fig:arb_spam_timeseries}B).
For off-chain discovery and on-chain evaluation bots, cumulative priority fees come predominantly from successful arbitrages despite spam transactions being roughly three times more frequent (cf. \Cref{tab:bot_label_activity}). This reflects that their high-value successful captures pay large tips and account for a disproportionate share of cumulative fees.
For on-chain discovery, by contrast, spam contributes a large share of cumulative priority fees because the architecture generates so many failed attempts.

\Cref{fig:priority_fee_composition}B shows a post-Flashblocks shift in arbitrage-related priority-fee revenue away from spam and toward successful arbitrages. Before Flashblocks, spam accounts for more than 60\% of weekly priority fees; afterward, it falls to roughly 40--45\%, even during the September entrant surge when spam transactions nearly double relative to pre-Flashblocks weeks (cf. \Cref{fig:arb_spam_timeseries}C). The successful-arbitrage share is correspondingly higher after Flashblocks, but remains opportunity-driven: it rises during high-volume periods, especially around the October episode, and declines as aggregate arbitrage volume falls (cf. \Cref{fig:arb_spam_timeseries}B). Thus, spam continues to contribute to nontrivial protocol revenue, but it becomes a less dominant component after the ordering granularity changes.



The fee-ramp period is harder to interpret because policy changes overlap with low-activity periods and a mid-ramp opportunity shock (cf. \cref{fig:arb_spam_timeseries}B). Even so, the successful-arbitrage share rises following the 500K step, despite volume contraction, consistent with selection against low-value, high-attempt-intensity bots. The pattern becomes less clear around the 1M and 2M steps, when the opportunity shock temporarily revives on-chain discovery.

\parhead{Blockspace efficiency.}
Despite protocol revenue contributions, failed attempts impose a burden on blockspace. \Cref{tab:gas_share} documents the gas usage by architecture as a share of total Base gas over the sample period. Although on-chain discovery bots only capture roughly half of successful arbitrages of the other two architectures (cf. \Cref{tab:bot_label_activity}), their arbitrages consume comparable blockspace because they are heavier (cf. \Cref{fig:gas_used_cdf}). Successful arbitrage transactions across all architectures consume only 1.31\% of total gas on Base. The blockspace burden is instead driven almost entirely by spam, especially on-chain discovery spam, which accounts for roughly 20\% of total Base gas.


\begin{table}[t]
\centering
\small
\setlength{\tabcolsep}{4pt}
\begin{tabular}{lccc}
\toprule
Architecture
& \makecell{Arbitrage gas share}
& \makecell{Spam gas share}
& \makecell{Total gas share} \\
\midrule
Off-chain discovery & 0.47\% & 0.99\% & 1.45\% \\
On-chain evaluation  & 0.41\% & 0.34\% & 0.75\% \\
On-chain discovery  & 0.44\% & 19.59\% & 20.04\% \\
\midrule
All bots        & 1.31\% & 20.92\% & 22.23\% \\
\bottomrule
\end{tabular}%
\caption{Successful arbitrage, spam, and total gas usage share out of total Base gas by architecture.}
\label{tab:gas_share}
\end{table}

\begin{figure}[t]
\centering
\includegraphics[
  width=0.8\linewidth
]{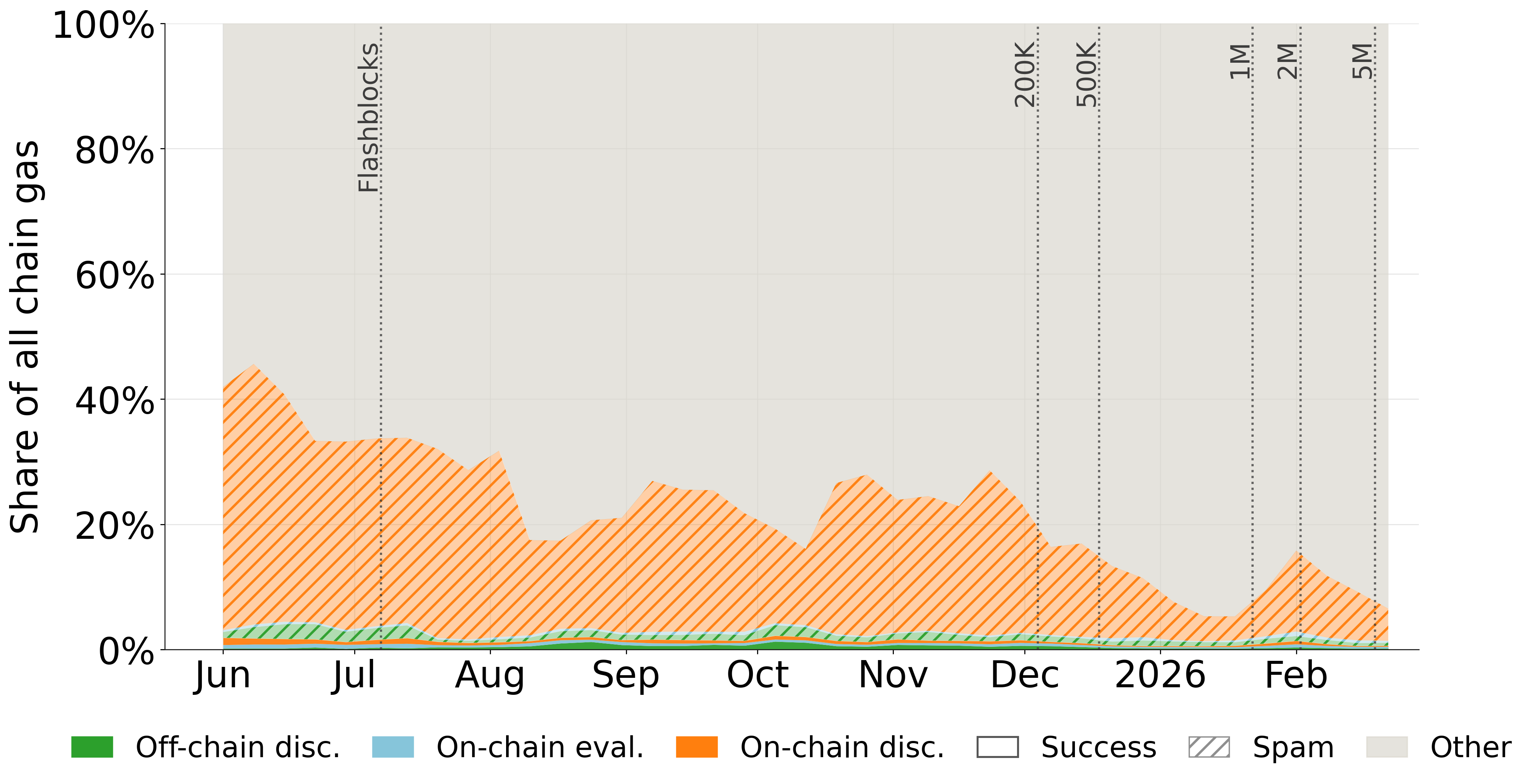}
\caption{Weekly arbitrage and spam gas share out of total Base gas by architecture.}
\label{fig:weekly_gas_share}
\end{figure}

\Cref{fig:weekly_gas_share} decomposes weekly gas usage into arbitrage-related gas and other chain activity. 
Before Flashblocks, on-chain discovery spam alone consumes roughly one-third to nearly one-half weekly gas. Flashblocks reduce this burden, although not immediately. As shown in \Cref{sec:flashblocks}, surviving and entering on-chain discovery bots become more attempt-intensive in the four-week post window. The gas share falls more clearly after that window to 15.2\%, as a few extremely intensive bots disappear or downscale (cf. \Cref{appendix:protocol_result}).

Later episodes show the same distinction between architecture selection and opportunity shocks. During the September surge and the November rebound period, spam transaction counts rise substantially, but spam does not return to its early-sample blockspace share. Through the fee ramp, spam gas share declines further, consistent with higher fees making low-yield probing harder to sustain. We interpret this cautiously because December and early January are lower-activity periods market-wide (cf. \Cref{fig:arb_spam_timeseries}B), and the late-January opportunity shock temporarily raises spam gas even under a higher fee floor.



Taken together, Base's configuration changes improve blockspace efficiency in a narrow but important sense. Spam continues to occur, while its share of chain gas falls after Flashblocks and declines further through the fee-ramp period, aside from temporary opportunity shocks.

\section{Concluding Discussion}
Our paper shows that arbitrage competition on high-throughput blockchains involves economically distinct search architectures with different responses to protocol design. A reduced-form model studies how ordering granularity, fee floors, and opportunity-access shocks affect the relative viability of targeted and probabilistic search. Empirical evidence from Base supports these comparative statics and traces how protocol changes reshape bot composition, execution behavior, priority-fee payments, and blockspace consumption.

The broader lesson is that protocol design affects not only how much arbitrage MEV is extracted, but also how it is pursued on-chain. Failed attempts are not merely isolated transaction-level inefficiencies; they arise from equilibrium search behavior under incomplete information, low attempt costs, and specific ordering environments. Opportunity shocks further show that protocol changes do not mechanically guarantee permanent reductions in failed attempts. 
The externality these attempts impose is blockspace congestion that displaces other users.
Understanding how search architectures emerge across protocol designs is important for reducing externalities without suppressing useful economic activity such as arbitrage. We conclude by outlining several open questions for future work.


\noindent\textbf{Does the search-architecture framework apply to other forms of MEV?}
Our analysis focuses on cyclic arbitrage, where the route discovery is visible in transaction traces. Other MEV forms may admit the same dichotomy. Liquidations can be targeted, when searchers wait for oracle updates against known positions, or probabilistic, when they repeatedly probe for liquidatable positions. Cross-domain arbitrage may similarly vary in how information is gathered and how transactions are submitted across domains. Future work could test whether these MEV forms exhibit analogous search architectures, how they respond to protocol changes, and what externalities they impose.

\noindent\textbf{Can the model's predictions be tested on other blockchains?}
A natural next step is to compare chains with different execution environments, fees, and information availability. Ethereum likely favors targeted search because failed attempts are costly, and mempool visibility enables off-chain simulation. Arbitrum may also favor targeted search, but for a different reason: its fast block time 
reduces stale-state risk, consistent with evidence that Arbitrum bots appear less probe-like than bots on Base and Optimism~\cite{ozan_lioba_paper}. Optimism is a closer Base comparison because both use the OP Stack~\cite{optimism2026opstack} and share similar ordering and configuration choices, including Flashblocks. Outside the Ethereum ecosystem, Solana offers a useful contrast as a high-throughput chain with active probabilistic searching~\cite{chen2023mevonsolana} and a different information environment, where transactions are forwarded directly to validators who may share pending-transaction information with searchers.

\noindent\textbf{How would priority-access mechanisms affect architecture competition?}
Our framework can also be extended to ordering mechanisms that allocate priority explicitly. Such mechanisms change the margin of competition from simply reaching the earliest ordering unit to deciding whether priority access is worth paying for and how aggressively to use it.
Arbitrum's Timeboost~\cite{arbitrum2026timeboostintro}, which sells priority access, tests one such margin: priority access may favor targeted searchers that can forecast and monetize latency advantages, while leaving probabilistic search viable only for operators with enough opportunity flow to justify the cost. Existing evidence that express-lane access is concentrated and that roughly 22\% of time-boosted transactions revert suggests that priority access alone does not eliminate spam~\cite{christof_timeboost}. Optimism's recent experiment with stake-based priority ordering~\cite{optimism2026stakebasedpriority} provides another useful setting for studying how stake requirements, first-come-first-served ordering, and priority-fee multipliers affect the relative viability of targeted and probabilistic search.

\noindent\textbf{How should spam mitigation differ across architectures?}
Failed attempts from different architectures call for different mitigation tools. Revert protection can reduce failed route-committed attempts~\cite{flashbots2025protectoverview,uniswap2025rollupboostunichain}, but it does not eliminate early-terminating targeted profitability checks or probabilistic probe spam, where transactions themselves are used to discover opportunities. Reducing probe-like spam requires improving information availability or reducing the incentive to use transactions as information-gathering devices, for example, through richer real-time information channels~\cite{quintus2024spamsearching}, protocol-run backrunning~\cite{osmosis2026protorev} and privacy-preserving backrunning auctions~\cite{flashbots2025mevshareintroduction}, and charging based on reserved execution gas~\cite{wang2026blockspacepressureanalysisspam}. On high-throughput chains, these mechanisms face a tight tradeoff between information freshness, auction efficiency, latency, and user protection.

\bibliography{reference}
\newpage
\appendix
\crefalias{section}{appendix}
\crefalias{subsection}{appendix}
\crefalias{subsubsection}{appendix} 
\section{Model Proofs and Regularity Conditions}
\label{sec:proofs}

This section provides the formal derivations omitted from \Cref{sec:model}. We first prove the within-regime equilibrium and comparative statics, then formalize the scan-footprint, fee-exposure, and opportunity-access extensions.

\subsection{Equilibrium proofs}
\begin{proof}[Proof of \Cref{prop:regime_eq}]
Fix a regime $r$ and suppress the dependence on $b$ by writing
\[
\eta := \eta_r, \qquad \kappa_O := \kappa_O(b), \qquad \kappa_C := \kappa_C(b).
\]

Suppose first that both types are active. Then the zero-profit condition for the on-chain architecture is
\[
\Pi_C^r(m_O,m_C;b)=0
\quad\Longrightarrow\quad
\eta\,\frac{\lambda_C V}{m_O+m_C}=K_C(b).
\]
Using the definition $\kappa_C=\lambda_C V/K_C(b)$, this implies
\[
m_O+m_C=\eta \kappa_C.
\]
Substituting this into the zero-profit condition for the off-chain architecture gives
\[
\eta\,\frac{\lambda_O V}{\eta \kappa_C}
+
(1-\eta)\,\frac{\lambda_O V}{m_O}
=
K_O(b).
\]
Since $\lambda_O V=\kappa_O K_O(b)$, we obtain
\[
\frac{\kappa_O K_O(b)}{\kappa_C}
+
(1-\eta)\,\frac{\kappa_O K_O(b)}{m_O}
=
K_O(b).
\]
Dividing by $K_O(b)$ and rearranging yields
\[
m_O
=
\frac{(1-\eta)\kappa_O}{1-\kappa_O/\kappa_C}.
\]
Because $\kappa_C>\kappa_O$, the denominator is strictly positive. Then
\[
m_C
=
\eta \kappa_C - m_O
=
\frac{\eta \kappa_C - \kappa_O}{1-\kappa_O/\kappa_C}.
\]
Hence the coexistence solution is feasible if and only if
\[
\eta \kappa_C > \kappa_O.
\]
Whenever this condition holds, the pair \((m_O,m_C)\) is pinned down uniquely, so the coexistence equilibrium on the branch with active off-chain discovery is unique.

Now consider the off-chain-only equilibrium. The off-chain architecture has zero profit because
\[
\Pi_O^r(\kappa_O,0;b)
=
\frac{\lambda_O V}{\kappa_O}-K_O(b)
=
K_O(b)-K_O(b)
=
0.
\]
At this profile, the on-chain architecture's payoff is
\[
\Pi_C^r(\kappa_O,0;b)
=
\eta\,\frac{\lambda_C V}{\kappa_O}-K_C(b)
=
K_C(b)\Bigl(\eta \frac{\kappa_C}{\kappa_O}-1\Bigr).
\]
Therefore the off-chain-only equilibrium is an equilibrium if and only if
\[
\eta \kappa_C \le \kappa_O.
\]
This proves both parts.
\end{proof}

\subsection{Blockspace-competitiveness channel}
\label{app:model_fb}
\begin{proof}[Proof of \Cref{thm:fb_competitiveness}]

Fix $b$ and suppress the dependence on $b$ by writing $\kappa_O=\kappa_O(b)$ and $\kappa_C=\kappa_C(b)$.
By \Cref{prop:regime_eq}, in any coexistence regime with parameter $\eta$,
\[
m_O^\ast(\eta)
=
\frac{(1-\eta)\kappa_O}{1-\kappa_O/\kappa_C},
\qquad
m_C^\ast(\eta)
=
\frac{\eta \kappa_C-\kappa_O}{1-\kappa_O/\kappa_C}.
\]
Because $\kappa_C>\kappa_O$, the denominator $1-\kappa_O/\kappa_C$ is strictly positive.

Now observe that
\[
\frac{\partial m_C^\ast(\eta)}{\partial \eta}
=
\frac{\kappa_C}{1-\kappa_O/\kappa_C}
>0,
\qquad
\frac{\partial m_O^\ast(\eta)}{\partial \eta}
=
-\frac{\kappa_O}{1-\kappa_O/\kappa_C}
<0.
\]
Hence \(m_C^\ast(\eta)\) is strictly increasing in \(\eta\), while \(m_O^\ast(\eta)\) is strictly decreasing in \(\eta\).

Since \(\eta_F<\eta_B\), it follows immediately that
\[
m_C^{\ast,F}(b)=m_C^\ast(\eta_F)<m_C^\ast(\eta_B)=m_C^{\ast,B}(b),
\]
and
\[
m_O^{\ast,F}(b)=m_O^\ast(\eta_F)>m_O^\ast(\eta_B)=m_O^{\ast,B}(b).
\]
This proves the result.
\end{proof}

\parhead{Selection on per-attempt scan footprint.}
Let on-chain-discovery architectures be indexed by an on-chain scan-footprint parameter $x$, where a larger $x$ means a more scan-intensive subtype with broader venue comparison or heavier pre-execution evaluation per attempt. For an on-chain subtype of scan-footprint $x$, define
\[
\kappa_C(x,b):=\frac{\lambda_C(x)V}{K_C(x,b)},
\]
where $\lambda_C(x)$ and $K_C(x,b)$ are the opportunity access rate and the per-opportunity cost for a subtype with footprint $x$, respectively.
Its regime-$r$ effective competitiveness is
\[
\Gamma_r(x,b):=\eta_r(x)\kappa_C(x,b),
\qquad r\in\{B,F\}.
\]

To compare on-chain discovery to the off-chain benchmark, define the relative competitiveness index
\[
R_r(x,b):=\frac{\Gamma_r(x,b)}{\kappa_O(b)}
=
\eta_r(x)\frac{\kappa_C(x,b)}{\kappa_O(b)}.
\]
An on-chain subtype is viable in regime $r$ if and only if
\[
R_r(x,b)>1.
\]
We use the weak-viability convention $R_r(x,b)\ge 1$ at the boundary; this convention is immaterial when the subtype distribution has no atom at $x_r^\ast(b)$.

We assume $R_r(x,b)$ is continuous and strictly decreasing in $x$: additional opportunity access from a broader scan footprint is not enough to offset lower same-slice competitiveness and higher cost. We also assume $R_F(x,b)<R_B(x,b)$ for all $x,b$, so finer ordering weakly lowers the competitiveness of each on-chain subtype.

Let the viability cutoff $x_r^\ast(b)$ be defined implicitly by
\[
R_r(x_r^\ast(b),b)=1.
\]
Then on-chain subtypes with
\[
x\le x_r^\ast(b)
\]
are viable in regime $r$, while more scan-intensive subtypes are not.

\begin{proof}[Proof of \Cref{prop:selection_on_complexity}]
By construction, an on-chain subtype is viable in regime $r$ if and only if
\[
R_r(x,b)>1.
\]
Since $R_r(x,b)$ is strictly decreasing in $x$, the set of viable on-chain subtypes is an interval of the form
\[
\{x:x\le x_r^\ast(b)\},
\]
where $x_r^\ast(b)$ is defined by
\[
R_r(x_r^\ast(b),b)=1.
\]
This proves part (1).

For part (2), fix $b$. Since
\[
R_F(x,b)<R_B(x,b)
\qquad\text{for all }x,
\]
the graph of $R_F(\cdot,b)$ lies everywhere below the graph of $R_B(\cdot,b)$. Therefore, the solution to
\[
R_F(x,b)=1
\]
occurs at a strictly lower value of $x$ than the solution to
\[
R_B(x,b)=1.
\]
Thus,
\[
x_F^\ast(b)<x_B^\ast(b).
\]
\end{proof}

\begin{corollary}[Finer ordering selects toward less scan-intensive survivors]
\label{cor:simpler_survivors}
Under the conditions of \Cref{prop:selection_on_complexity}, let $X$ be drawn from a fixed distribution over candidate on-chain subtypes, independent of the ordering regime, and suppose
\[
\Pr[X\le x_F^\ast(b)]>0.
\]
Let $z(X)$ be any strictly increasing observable scan-footprint measure, such as gas usage per attempt or the number of pre-swap state reads. Then
\[
\mathbb{E}[z(X)\mid X\le x_F^\ast(b)]
\le
\mathbb{E}[z(X)\mid X\le x_B^\ast(b)],
\]
with strict inequality whenever the distribution places positive mass on
\[
(x_F^\ast(b),x_B^\ast(b)].
\]
Thus, finer ordering shifts the surviving viable on-chain-discovery population toward a lower observable scan footprint.
\end{corollary}

\begin{proof}[Proof of \Cref{cor:simpler_survivors}]
Under Proposition~\ref{prop:selection_on_complexity}, the viable on-chain set in regime $r$ is
\[
\{x:x\le x_r^\ast(b)\}.
\]
Since
\[
x_F^\ast(b)<x_B^\ast(b),
\]
the viable set under finer ordering is a truncation from above of the viable set under full-block ordering.

Let $F_X$ denote the distribution of $X$. For any $t\le x_F^\ast(b)$,
\[
\Pr[X\le t\mid X\le x_F^\ast(b)]
=
\frac{F_X(t)}{F_X(x_F^\ast(b))}
\ge
\frac{F_X(t)}{F_X(x_B^\ast(b))}
=
\Pr[X\le t\mid X\le x_B^\ast(b)].
\]
Thus, the conditional distribution under finer ordering places more mass on lower-complexity subtypes. It first-order stochastically dominates the full-block conditional distribution in the direction of lower complexity.

Therefore, for any strictly increasing function $z(\cdot)$,
\[
\mathbb{E}[z(X)\mid X\le x_F^\ast(b)]
\le
\mathbb{E}[z(X)\mid X\le x_B^\ast(b)].
\]
The inequality is strict whenever the distribution places positive mass on
\[
(x_F^\ast(b),x_B^\ast(b)].
\]
\end{proof}

\subsection{Cost channel}
\label{app:model_fee}
\begin{proof}[Proof of \Cref{prop:fee_floor}]
By definition,
\[
\frac{\kappa_C(b)}{\kappa_O(b)}
=
\frac{\lambda_C V/K_C(b)}{\lambda_O V/K_O(b)}
=
\frac{\lambda_C}{\lambda_O}\cdot \frac{K_O(b)}{K_C(b)}.
\]
Taking logs gives
\[
\log \frac{\kappa_C(b)}{\kappa_O(b)}
=
\log\frac{\lambda_C}{\lambda_O}
+
\log K_O(b)
-
\log K_C(b).
\]
Differentiating with respect to $b$,
\[
\frac{d}{db}
\log \frac{\kappa_C(b)}{\kappa_O(b)}
=
\frac{K_O'(b)}{K_O(b)}
-
\frac{K_C'(b)}{K_C(b)}
=
\chi_O(b)-\chi_C(b).
\]
If $\chi_C(b)>\chi_O(b)$, then this derivative is negative, so
\[
\frac{d}{db}
\left(
\frac{\kappa_C(b)}{\kappa_O(b)}
\right)
<0.
\]
The coexistence condition can be written as
\[
\eta_r>\frac{\kappa_O(b)}{\kappa_C(b)}.
\]
Since $\kappa_C(b)/\kappa_O(b)$ falls with $b$, the right-hand side rises with $b$.
Therefore, the condition becomes harder to satisfy.
\end{proof}

The same fee-exposure logic applies within on-chain discovery. For subtype comparisons, the off-chain benchmark $\kappa_O(b)$ cancels from the relative competitiveness ratio, so the result can be stated equivalently in terms of $\Gamma$ or the relative competitiveness index $R$.

\begin{proposition}[Fee floors select against high-exposure on-chain subtypes]
\label{prop:fee_exposure_selection}
Fix an ordering regime $r$ and consider two on-chain discovery subtypes $j$ and $k$ with effective competitiveness
\[
\Gamma_{r,j}(b)=\eta_{r,j}\frac{\lambda_j V}{K_j(b)},
\qquad
\Gamma_{r,k}(b)=\eta_{r,k}\frac{\lambda_k V}{K_k(b)}.
\]
Assume $\eta_{r,j}$, $\eta_{r,k}$, $\lambda_j$, and $\lambda_k$ are fixed with respect to $b$. If
\[
\chi_j(b):=\frac{K_j'(b)}{K_j(b)}
>
\chi_k(b):=\frac{K_k'(b)}{K_k(b)},
\]
then
\[
\frac{d}{db}
\left(
\frac{\Gamma_{r,j}(b)}{\Gamma_{r,k}(b)}
\right)<0.
\]
Thus, as the fee floor rises, subtype $j$ loses relative competitiveness to subtype $k$.
\end{proposition}

\begin{proof}[Proof of \Cref{prop:fee_exposure_selection}]
Taking logs,
\[
\log \frac{\Gamma_{r,j}(b)}{\Gamma_{r,k}(b)}
=
\log\frac{\eta_{r,j}\lambda_j}{\eta_{r,k}\lambda_k}
+
\log K_k(b)
-
\log K_j(b).
\]
Differentiating with respect to $b$ gives
\[
\frac{d}{db}
\log \frac{\Gamma_{r,j}(b)}{\Gamma_{r,k}(b)}
=
\chi_k(b)-\chi_j(b).
\]
If $\chi_j(b)>\chi_k(b)$, this derivative is negative. Since $\Gamma_{r,j}(b)/\Gamma_{r,k}(b)>0$, the ratio itself is strictly decreasing in $b$.
\end{proof}

\subsection{Opportunity access shock}
\label{app:model_shock}
Fix a policy environment $(r,b)$, and suppose off-chain discovery remains active. Define
\[
\bar{\lambda}_C(r,b)
:=
\frac{\lambda_O K_C(b)}{\eta_r K_O(b)}.
\]
On-chain discovery is viable in environment $(r,b)$ if and only if $\lambda_C > \bar{\lambda}_C(r,b)$. The threshold is increasing in both policy frictions:
\[
\eta_F<\eta_B
\Rightarrow
\bar{\lambda}_C(F,b)>\bar{\lambda}_C(B,b),
\qquad
\chi_C(b)>\chi_O(b)
\Rightarrow
\frac{\partial \bar{\lambda}_C(r,b)}{\partial b}>0.
\]
Consider a baseline opportunity environment with $\lambda_C^0$ and a shock environment with $\lambda_C^1$. If
\[
\lambda_C^0
\le
\bar{\lambda}_C(r,b)
<
\lambda_C^1,
\]
Then, on-chain discovery is not viable in the baseline environment but becomes viable in the shock environment, holding policy fixed.

\begin{proof}[Proof of \Cref{prop:launch}]
The viability condition for on-chain discovery is
\[
\eta_r\kappa_C(b)>\kappa_O(b).
\]
Using
\[
\kappa_C(b)=\frac{\lambda_C V}{K_C(b)}
\qquad\text{and}\qquad
\kappa_O(b)=\frac{\lambda_O V}{K_O(b)},
\]
this condition becomes
\[
\eta_r\frac{\lambda_C V}{K_C(b)}
>
\frac{\lambda_O V}{K_O(b)}.
\]
Canceling $V$ and rearranging yields
\[
\lambda_C
>
\frac{\lambda_O K_C(b)}{\eta_r K_O(b)}
=
\bar{\lambda}_C(r,b).
\]
This proves the threshold characterization.

Since
\[
\bar{\lambda}_C(r,b)
=
\frac{\lambda_O K_C(b)}{\eta_r K_O(b)},
\]
a lower value of $\eta_r$ raises the threshold. Therefore $\eta_F<\eta_B$ implies
\[
\bar{\lambda}_C(F,b)>\bar{\lambda}_C(B,b).
\]

For the fee-floor comparative static, take logs:
\[
\log \bar{\lambda}_C(r,b)
=
\log \lambda_O
+
\log K_C(b)
-
\log \eta_r
-
\log K_O(b).
\]
Holding $r$ fixed and differentiating with respect to $b$ gives
\[
\frac{\partial}{\partial b}
\log \bar{\lambda}_C(r,b)
=
\frac{K_C'(b)}{K_C(b)}
-
\frac{K_O'(b)}{K_O(b)}
=
\chi_C(b)-\chi_O(b).
\]
Hence, if $\chi_C(b)>\chi_O(b)$, then
\[
\frac{\partial}{\partial b}\log \bar{\lambda}_C(r,b)>0,
\]
which implies
\[
\frac{\partial \bar{\lambda}_C(r,b)}{\partial b}>0.
\]

If
\[
\lambda_C^0
\le
\bar{\lambda}_C(r,b)
<
\lambda_C^1,
\]
then the viability condition fails in the baseline environment and holds in the launch environment. Hence on-chain discovery is not viable at $s=0$ but becomes viable at $s=1$.

Finally, on the coexistence branch,
\[
m_C^{\ast,r}
=
\frac{\eta_r\kappa_C-\kappa_O}{1-\kappa_O/\kappa_C},
\qquad
m_O^{\ast,r}
=
\frac{(1-\eta_r)\kappa_O}{1-\kappa_O/\kappa_C}.
\]
Differentiating with respect to $\kappa_C$ gives
\[
\frac{\partial m_C^{\ast,r}}{\partial \kappa_C}
=
\frac{\eta_r(\kappa_C-\kappa_O)^2+\kappa_O^2(1-\eta_r)}
{(\kappa_C-\kappa_O)^2}
>0,
\qquad
\frac{\partial m_O^{\ast,r}}{\partial \kappa_C}
=
-\frac{(1-\eta_r)\kappa_O^2}{(\kappa_C-\kappa_O)^2}
<0.
\]
Since $\partial \kappa_C/\partial \lambda_C=V/K_C(b)>0$, the comparative statics with respect to $\lambda_C$ follow.

\end{proof}

\section{Trace Classifier: Implementation and Limitations}
\label{app:trace_classifier}

\subsection{Classifier implementation details}
\label{app:classifier_implement}
\begin{algorithm}[H]
\DontPrintSemicolon
\SetAlgoLined
\KwIn{Call trace $T$ of a sampled successful arbitrage transaction}
\KwOut{Transaction architecture label $\ell$}

Flatten $T$ into an ordered call list and partition non-root calls into depth-1
subtrees, interpreted as route-attempt subtrees\;
Construct the bot ecosystem from the root contract, delegatecall targets, and
bot-controlled transfer recipients\;

\ForEach{route-attempt subtree}{
    Detect direct-entry swap subtrees, where execution begins without prior venue reads\;
    Otherwise split the subtree at its first successful swap-like venue call\;
    Record pre-swap pool-state reads, \texttt{balanceOf} pool targets, revert-quote calls,
    and venues used in successful execution\;
    Identify helper contracts whose internal calls reach venues later used in execution\;
}

Let $\mathcal{R}$ be direct pool-state-read venues and balance-read targets before the first successful swap\;
Let $\mathcal{Q}$ be pre-swap revert-quoted venues\;
Let $\mathcal{H}$ be helper-resolved venues whose internal calls reach executed venues\;

Remove from broad-scan evidence bot-internal addresses, helper-resolved venues,
and revert-quoted venues\;
If the executed route uses a singleton pool manager, treat \texttt{balanceOf}
targets from the successful route attempt as route-local because they need not
coincide with the PoolManager address\;
Let $\mathcal{X}$ be venues in $\mathcal{R}$ that are not executed, helper-resolved, revert-quoted, or route-local singleton balance targets\;

\uIf{$\mathcal{R}=\emptyset$ and $\mathcal{Q}=\emptyset$}{
    $\ell\leftarrow\texttt{OFFCHAIN\_DIRECT}$\tcp*{no meaningful pre-swap state read}
}
\uElseIf{all pre-swap evidence is helper-resolved, route-local singleton evidence,
or revert-quoted on the executed route}{
    $\ell\leftarrow\texttt{OFFCHAIN\_DIRECT}$\tcp*{on-chain route pricing with off-chain simulation}
}
\uElseIf{$\mathcal{X}\neq\emptyset$}{
    $\ell\leftarrow\texttt{ONCHAIN\_BROAD\_SCAN}$\tcp*{read-based venue expansion}
}
\uElseIf{pre-swap reads are confined to executed venues}{
    $\ell\leftarrow\texttt{ONCHAIN\_EVAL}$\tcp*{route-confined state evaluation}
}
\uElseIf{non-executed venues appear only through revert quotes or helper calls}{
    $\ell\leftarrow\texttt{ONCHAIN\_EVAL}$\tcp*{weak route comparison evidence}
}
\Else{
    $\ell\leftarrow\texttt{OFFCHAIN\_DIRECT}$\tcp*{conservative fallback}
}

\Return{$\ell$}\;

\caption{Implementation of trace-level classifier for successful arbitrages. The algorithm infers search architecture from trace-visible route discovery within arbitrage traces: direct execution, route-confined evaluation, or broad on-chain scanning.}
\label{alg:tx_classifier}
\end{algorithm}

\subsection{Empirical assessment of classifier limitations}
\label{app:classifier_robustness}
In this section, we quantify the empirical importance of the main limitations of our classifier and show that their impact on the final bot-week architecture labels is limited.

First, Uniswap V4-style singleton \texttt{PoolManager} designs can obscure pool identity because multiple logical pools are accessed through a common DEX contract address. This can cause some genuine broad-scan behavior to be absorbed into \texttt{ONCHAIN\_EVAL} rather than \texttt{ONCHAIN\_BROAD\_SCAN}, so the bias is one-sided and would potentially undercount on-chain discovery. 

In our dataset, however, the margin appears limited. Among the 141,751 transactions labeled \texttt{ONCHAIN\_EVAL}, 86.5\% contain no singleton-manager reads. As a stricter diagnostic, we flag \texttt{ONCHAIN\_EVAL} transactions in which the distinct singleton-manager read count (with distinct calldata) exceeds its executed-venue count. Only 2,190 transactions in our sample fall under this category, representing 1.55\% of \texttt{ONCHAIN\_EVAL} transactions and 0.61\% of the full classified sample. At the bot-week level, such transactions appear in only 207 on-chain evaluation bot-weeks or 1.91\% of the full panel. 

In addition, we note that multiple distinct singleton-manager reads do not necessarily imply that the bot is evaluating multiple venues --- the bot could be evaluating the same venue multiple times with different inputs, meaning the numbers above are upper bounds. We therefore interpret singleton-manager opacity as a limited but non-zero source of downward bias in measured on-chain discovery.


Second, weekly bot labels are inferred from a limited sample of up to 50 successful arbitrage transactions, so an off-chain discovery or on-chain discovery bot can be misclassified as on-chain evaluation if the sample misses its defining transaction type entirely. Under a simple model, the miss probability is $(1-p)^n$, where $p$ is the within-week share of the defining type and $n$ is the sample size. Therefore, the empirical question is whether discovery bot-weeks typically have very low defining-type shares. 

We find that it is not the case. Across bot-weeks assigned to off-chain discovery, the median \texttt{OFFCHAIN\_DIRECT} transaction share is 1.00, and the 25th percentile is also 1.00; across bot-weeks assigned to on-chain discovery, the median \texttt{ONCHAIN\_BROAD\_SCAN} transaction share is 0.96. Sampling uncertainty is therefore concentrated in a thin tail of borderline cases. Among bot-weeks for which the classifier sampled the maximum 50 transactions, only 112 off-chain-discovery bot-weeks and 24 on-chain-discovery bot-weeks contain exactly one defining transaction in the 50-transaction sample. 
The plug-in calculation implies about 81 potentially missed on-chain discovery or off-chain discovery labels. The relevant high-activity comparison set contains 4,472 bot-weeks assigned to either off-chain or on-chain discovery with 50 sampled transactions, so the estimated miss rate is 1.81\% within that set. Relative to the full panel of 10,828 bot-weeks, the same 81 cases equal 0.75\%.
This suggests that the main empirical results are driven by confidently labeled bot architectures rather than by marginal bot-weeks near the classification threshold.

\section{Additional Empirical Results}
\label{app:additional_empirics}




\subsection{Additional results on Flashblocks}
\label{appendix:flashblocks_result}

\parhead{Modest changes outside on-chain discovery.}
\Cref{tab:flashblocks_offchain_architecture} reports bot-level architecture metrics and transaction-level footprint measures for off-chain discovery and on-chain evaluation bots around Flashblocks. In contrast to the sharp selection within on-chain discovery, both groups exhibit only modest changes in observable architecture.

\begin{table}[H]
\centering
\resizebox{\linewidth}{!}{%
\begin{tabular}{@{}lrrrrrrr@{}}
\toprule
& & \multicolumn{3}{c}{Bot-level architecture}
& \multicolumn{3}{c}{Transaction-level footprint} \\
\cmidrule(lr){3-5}\cmidrule(l){6-8}
Bot label & Bots
& \makecell{Median \\ calldata length}
& \makecell{Average venue \\ commitment}
& \makecell{Median venue \\ commitment}
& \makecell{Median \\ pre-reads}
& \makecell{Median \\ arbitrage tx gas}
& \makecell{Median \\ spam tx gas} \\
\midrule
Off-chain disc. (pre $\rightarrow$ post)
& 134 $\rightarrow$ 207
& 415 $\rightarrow$ 356
& 0.736 $\rightarrow$ 0.622
& 0.999 $\rightarrow$ 0.860
& 0 $\rightarrow$ 0
& 351K $\rightarrow$ 319K
& 265K $\rightarrow$ 275K \\
On-chain eval. (pre $\rightarrow$ post)
& 86 $\rightarrow$ 113
& 190 $\rightarrow$ 193
& 0.730 $\rightarrow$ 0.736
& 1.000 $\rightarrow$ 1.000
& 3 $\rightarrow$ 4
& 293K $\rightarrow$ 293K
& 43.5K $\rightarrow$ 49.6K \\
\bottomrule
\end{tabular}
}
\caption{Architecture metrics for off-chain discovery and on-chain evaluation bots around Flashblocks. Bot-level architecture metrics first summarize each bot over its labeled bot-weeks within the relevant window then summarize across bots. Transaction-level footprint metrics use sampled arbitrage traces for pre-reads, full successful arbitrage transactions for arbitrage gas, and bot-week spam medians for spam gas.}
\label{tab:flashblocks_offchain_architecture}
\end{table}

\parhead{Flashblocks Window Robustness}
The main Flashblocks event study uses a symmetric four-week window around the July 7, 2025 activation, excluding the event week. This window gives enough support to study entry, exit, and same-label survivor behavior, but it also raises the concern that the result may be influenced by the broader pre-event decline or by later post-event weeks. We therefore repeat the analysis using a narrower two-week window before and after Flashblocks. 

\begin{table}[H]
\centering
\resizebox{\linewidth}{!}{%
\begin{tabular}{lcc}
\toprule
Metric & 4-week main window & 2-week robustness window \\
\midrule
Window dates & Jun. 8--Jul. 5 $\rightarrow$ Jul. 13--Aug. 9 & Jun. 22--Jul. 5 $\rightarrow$ Jul. 13--Jul. 26 \\
Active bots & 983 $\rightarrow$ 462 & 547 $\rightarrow$ 314 \\
Average weekly arbitrage txs & 581{,}685 $\rightarrow$ 487{,}262 & 508{,}700 $\rightarrow$ 533{,}845 \\
On-chain discovery bots & 782 $\rightarrow$ 157 & 387 $\rightarrow$ 112 \\
On-chain discovery bot share & 63.1\% $\rightarrow$ 25.1\% & 60.7\% $\rightarrow$ 28.4\% \\
Off-chain discovery bot share & 23.4\% $\rightarrow$ 46.7\% & 25.5\% $\rightarrow$ 41.9\% \\
On-chain eval bot share & 13.4\% $\rightarrow$ 28.2\% & 13.8\% $\rightarrow$ 29.6\% \\
\midrule
Selected-out on-chain discovery bots & 741 of 782 (94.8\%) & 348 of 387 (89.9\%) \\
On-chain discovery entrants & 120 & 77 \\
Same-label on-chain discovery survivors & 37 & 35 \\
\midrule
Selected-out bot-level median calldata length & 36 & 36 \\
Selected-out average bot venue commitment & 0.044 & 0.048 \\
Selected-out median bot venue commitment & 0.000 & 0.000 \\
Selected-out tx-level median pre-reads & 57 & 46 \\
Survivor bot-level median calldata length & 163 $\rightarrow$ 302 & 163 $\rightarrow$ 512 \\
Survivor average bot venue commitment & 0.345 $\rightarrow$ 0.340 & 0.334 $\rightarrow$ 0.329 \\
Survivor median bot venue commitment & 0.000 $\rightarrow$ 0.000 & 0.000 $\rightarrow$ 0.000 \\
Survivor tx-level median pre-reads & 39 $\rightarrow$ 12 & 45 $\rightarrow$ 16 \\
\bottomrule
\end{tabular}
}
\caption{Flashblocks robustness to using a two-week event window. The main window compares June 8--July 5, 2025 to July 13--August 9, 2025; the robustness window compares June 22--July 5, 2025 to July 13--July 26, 2025. The Flashblocks activation week, July 6--July 12, 2025, is excluded in both cases. Calldata length and venue commitment are computed at the bot level; pre-reads are computed at the sampled-transaction level.}
\label{tab:flashblocks_2w_robustness}
\end{table}

\Cref{tab:flashblocks_2w_robustness} presents the comparison. The shorter window preserves the main empirical pattern. Active bots fall from 547 to 314, a 42.6\% decline, while average weekly arbitrage transactions increase slightly from 508{,}700 to 533{,}845. The contraction is again concentrated in on-chain discovery: active on-chain discovery bots fall from 387 to 112, and their active-bot share falls from 60.7\% to 28.4\%. At the same time, off-chain discovery rises from 25.5\% to 41.9\% of active bots, and on-chain evaluation rises from 13.8\% to 29.6\%. Thus, even in the more local window, the decline in on-chain discovery is not mechanically explained by a collapse in realized arbitrage activity.

The selection and adaptation margins also remain in the shorter window. In the two-week window, 348 of 387 pre-window on-chain discovery bots are selected out, compared with 741 of 782 in the four-week window. These selected-out bots continue to look like broad scanners: their bot-level median calldata length is only 36 bytes, their median bot venue commitment is zero, and their transaction-level median pre-read count is 46. Same-label on-chain discovery survivors also become less search-intensive after Flashblocks, with transaction-level median pre-reads falling from 45 to 16. Their bot-level venue commitment does not rise for the typical survivor: the median remains zero, and the average is nearly unchanged at 0.334 to 0.329. Thus, the robust same-bot adaptation signal is lower pre-execution reading, not a broad increase in route commitment. The two-week window therefore supports the same interpretation as the main analysis: Flashblocks primarily operates through selection against search-intensive on-chain discovery bots, while adaptation among surviving on-chain discovery bots is present but secondary.

\subsection{Additional results on the September surge}
\label{appendix:sep_surge_result}

\parhead{Short-lived pattern in transaction intensity.}\Cref{tab:september_entrant_wave} shows that this short-lived pattern is also visible in transaction intensity. Among transient on-chain discovery bots that are both first-seen and last-seen in the surge window, the median arbitrage transaction count is only 29 over weeks where the bot is labeled as on-chain discovery. By contrast, the non-transient on-chain discovery bots have a median of 262.5 surge-window arbitrage transactions and 4,196 full-sample arbitrage transactions. The on-chain evaluation spike has the same entrant-wave flavor, but is even thinner on the transaction count: short-lived on-chain evaluation bots execute only 2 surge-window arbitrage transactions at the median, compared with 682 for the remaining on-chain evaluation bots. The surge episode is therefore bimodal: a small number of serious entrants and a much larger mass of short-lived experiments. 

\begin{table}[H]
\centering
\resizebox{\linewidth}{!}{%
\begin{tabular}{lrrrrr}
\toprule
Bot label
& \makecell{Unique \\ bots}
& \makecell{First seen \\ in window}
& \makecell{Last seen \\ in window}
& \makecell{First and last \\ seen in window}
& \makecell{Median arbs in window \\ transient / non-transient} \\
\midrule
On-chain disc. & 504 & 482 (95.6\%) & 478 (94.8\%) & 464 (92.1\%) & 29 / 262.5 \\
On-chain eval. & 268 & 211 (78.7\%) & 204 (76.1\%) & 192 (71.6\%) & 2 / 682 \\
Off-chain disc. & 220 & 104 (47.3\%) & 110 (50.0\%) & 78 (35.5\%) & 15 / 138.5 \\
\bottomrule
\end{tabular}
}
\caption{Entrant-wave composition by architecture during the September surge window. ``First seen'' and ``last seen'' are defined over all labeled bot-week activity, not only within the architecture label. Transaction counts are successful arbitrages in September bot-weeks per architecture label.}
\label{tab:september_entrant_wave}
\end{table}

\parhead{November rebound.}
The later November rebound has a different structure from the September AVNT/MIRROR surge. We define the rebound as the five weekly windows beginning October 26 through November 23, 2025, when on-chain discovery spam averages 38.4M transactions per week. The episode coincides with sustained x402-ecosystem volatility on Base \cite{x402}: PING, VIRTUAL, AIXBT, CLANKER, and AVNT rally sharply around October 24, followed by an x402-related correction on November 11; SoSoValue's SoDEX launches on October 28; AVNT remains active; and bridged or newly deployed tokens such as bAP3X enter Base liquidity venues. 

The resulting opportunity set is distributed across many token and venue pairs rather than centered on a single launch. Consistent with this broader structure, off-chain discovery captures 57.1\% of November volume, while on-chain discovery captures 25.9\% despite producing 97.1\% of labeled spam. Thus, the November rebound is not a pure on-chain discovery entrant swarm. It is a high-opportunity period in which route-committed bots capture much of the value, while on-chain discovery produces most of the failed search externality.

Compared with the September surge, the November rebound is less of a transient entrant wave and more of a persistent-bot episode. In November, on-chain discovery contains 179 active bots, but the 48 non-transient bots generate most of the activity: 115.2M spam transactions versus 76.7M from 131 transient bots, 74.0\% of on-chain discovery arbitrage transactions, and 57.1\% of on-chain discovery volume. Across all labels, bots already observed before the rebound produce 55.2\% of labeled November volume, while non-transient bots, including incumbents and newly entering persistent bots, produce 76.0\%. 

The route evidence is consistent with this interpretation. Major on-chain discovery routes span USDC/TRUST, SOSO/USDC/ETH, bAP3X/AERO, PING/USDC, and VIRTUAL/ETH. Thus, November shows that spam can rebound after Flashblocks when a broad opportunity portfolio raises the value of search, but the rebound is more incumbent, concentrated, and off-chain-heavy than a narrow token-launch swarm.


\subsection{Additional results on fee floor escalation}
\label{appendix:fee_result}

\parhead{Value-per-attempt comparison.}
\Cref{tab:fee_step_value_per_attempt} reports the aggregate USD volume per attempt for on-chain discovery cohorts around each minimum-base-fee increase. At 200K and 5M steps, surviving bots are capturing more value per attempt in the post-step window than in the pre-step window, consistent with within-bot pruning of low-yield attempts or shifts toward higher-value opportunities. The 500K step moves in the opposite direction (0.50 $\rightarrow$ 0.39), but a market-wide contraction confounds this comparison in late December --- total arbitrage transactions fall from 512K to 332K, and total volume falls from \$66M to \$36M --- so we treat it as confounded.
We note that the entrant rows do not necessarily contradict the selection mechanism. At 5M, entrants have lower aggregate USD per attempt than surviving incumbents (0.65 versus 1.31). These bots are observed only in the post-step week, so they may include exploratory entry or short-lived responses to local opportunities rather than stable post-fee-floor participants. The selection mechanism is therefore more directly identified in survivor adjustment and subsequent persistence than in one-week entry alone.

\begin{table}[t]
\centering
\begin{tabular}{llrr}
\toprule
Fee step & Cohort & Bots & \makecell{Aggregate \\ USD volume / attempt} \\
\midrule
200K & Selected out (pre-only) & 27 & 0.27 \\
200K & Survivors (pre $\rightarrow$ post) & 15 & 0.36 $\rightarrow$ 0.87 \\
200K & Entrants (post-only) & 22 & 0.44 \\
\midrule
500K & Selected out (pre-only) & 18 & 0.76 \\
500K & Survivors (pre $\rightarrow$ post) & 19 & 0.50 $\rightarrow$ 0.39 \\
500K & Entrants (post-only) & 9 & 1.32 \\
\midrule
1M & Selected out (pre-only) & 9 & 1.42 \\
1M & Survivors (pre $\rightarrow$ post) & 20 & 0.69 $\rightarrow$ 2.68 \\
1M & Entrants (post-only) & 10 & 60.73 \\
\midrule
2M & Selected out (pre-only) & 13 & 91.0 \\
2M & Survivors (pre $\rightarrow$ post) & 17 & 1.98 $\rightarrow$ 0.32 \\
2M & Entrants (post-only) & 16 & 2.45 \\
\midrule
5M & Selected out (pre-only) & 18 & 1.61 \\
5M & Survivors (pre $\rightarrow$ post) & 15 & 0.46 $\rightarrow$ 1.31 \\
5M & Entrants (post-only) & 20 & 0.65 \\
\bottomrule
\end{tabular}
\caption{Aggregate USD per attempt for on-chain discovery cohorts around each minimum-base-fee increase. Each step uses a one-week pre-window and a one-week post-window. USD volume per attempt is computed as arbitrage volume divided by total attempts, summed within each cohort before taking the ratio.}
\label{tab:fee_step_value_per_attempt}
\end{table}

\parhead{Mid-ramp opportunity shock.}
The 1M and 2M fee steps overlap a smaller opportunity-driven on-chain discovery episode from late January to mid-February (cf. \Cref{fig:arb_spam_timeseries}). During this period, active on-chain discovery bots rise from a baseline of roughly 25 to 47 at the peak, and several memecoin token families enter Base arbitrage paths, including ROLL, MOLT, CLAWNCH, CLAWD, and KellyClaude, with 90\% of their full-sample arbitrage activity occurring in this short window.

This shock explains why the 1M and 2M fee-step cohort metrics do not cleanly identify fee-floor selection in \Cref{tab:fee_step_value_per_attempt}. At the 1M step, the post-step window coincides with the shock peak. Entrants in this window have an aggregate USD per attempt of 60.73, consistent with bots entering to capture newly available high-value token opportunities. Surviving bots' value per attempt also rises sharply, from 0.69 to 2.68, as incumbents capture part of the same shock flow. At the 2M step, the comparison reverses mechanically: the pre-step window is still near the shock peak, while the post-step window occurs after the shock has subsided. The cohort selected out at 2M has an aggregate USD per attempt of 91.0 because it is measured under shock conditions, not because selected-out bots are generally high-quality steady-state operators. Surviving bots' value per attempt then falls from 1.98 to 0.32 as the opportunity dissipates. We therefore interpret the 1M and 2M rows as tracing a temporary opportunity cycle, rather than as clean fee-floor treatment effects.

The episode resembles the September surge qualitatively but is much smaller and more incumbent-dominated. It contains 64 unique on-chain discovery bots, and only 31 are first-and-last-seen entrants. Activity is also more concentrated: the top 10 on-chain discovery bots capture 84.5\% of arbitrage transactions, while new entrants contribute only 21.3\%. This contrast is consistent with a higher fee-floor viability threshold. Under the zero-fee-floor September regime, marginal trier bots could cheaply enter around short-lived token-launch opportunities. Under the January--February 1M--2M WEI fee floor, each speculative attempt is more expensive, so the revival is dominated by established high-volume operators rather than a broad entrant swarm.


\subsection{Additional results on chain performance}
\label{appendix:protocol_result}

\parhead{Priority fee results.}
\Cref{fig:prio_fee_spam_dist} presents the distribution of the average priority fee per spam transaction within a bot-week. Off-chain discovery bots have the highest full-sample average priority fee per spam transaction, at \SI{1.06e-6} ETH, compared with \SI{5.66e-7} ETH for on-chain evaluation and \SI{5.29e-7} ETH for on-chain discovery, indicating that many failed attempts are economically meaningful arbitrage attempts rather than arbitrary background noise. After Flashblocks, the on-chain discovery spam-priority distribution also shifts closer to the other architectures, consistent with selection away from the lowest-willingness-to-pay search. 

\Cref{fig:prio_fee_week26_vs_week39} presents the priority fee distribution of successful arbitrages in the pre-ramp endpoint week and post-ramp endpoint week. The median priority fee of on-chain discovery arbitrages rises from \SI{4.58e-7} ETH in the pre-ramp endpoint week to \SI{2.69e-6} ETH in the post-ramp endpoint week. Because we observe only one post-ramp week, we interpret this as an equilibrium outcome combining selection, opportunity value, and bidding behavior, not as a mechanical effect of the fee change.
\begin{figure}[H]
    \centering
    \includegraphics[width=1\linewidth]{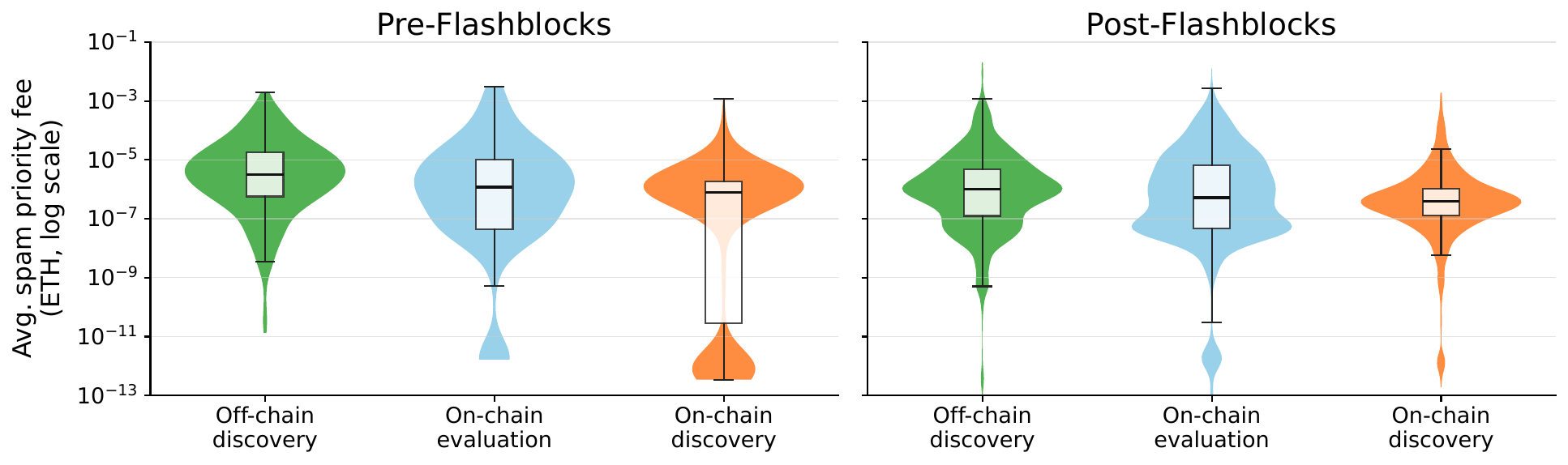}
    \caption{Violin plot of average priority fee per spam tx within a bot-week pre- and post-Flashblocks by architecture.}
    \label{fig:prio_fee_spam_dist}
\end{figure}

\begin{figure}[H]
    \centering
    \includegraphics[width=1\linewidth]{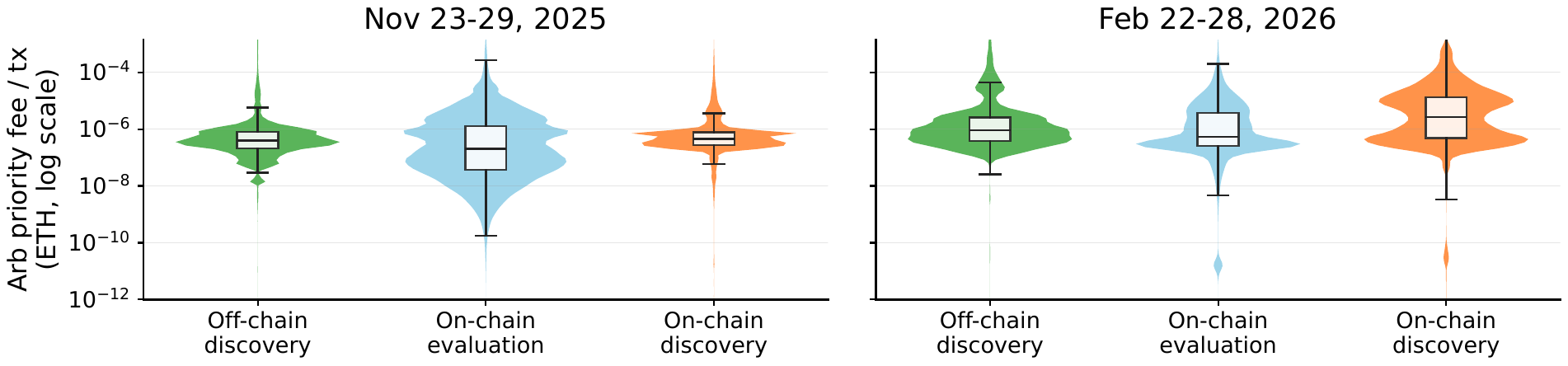}
    \caption{Average priority fee distributions for successful arbitrage transactions by architecture in the pre- and post-ramp endpoint weeks.}
    \label{fig:prio_fee_week26_vs_week39}
\end{figure}

\parhead{Blockspace results.}
As shown in \Cref{sec:flashblocks}, Flashblocks does not immediately reduce aggregate spam or spam gas in the four-week post-window because surviving and entering on-chain discovery bots become more attempt-intensive. The reduction appears immediately afterward. In the week beginning August 10, 2025, spam transactions fall by only 25\% relative to the four-week post-window average, from roughly 20M to 15.38M, but spam gas falls nearly 50\%, from 4.6T to 2.4T, and spam gas share falls from 30.6\% to 16.1\%.

The break is not simply a broad decline in the number of on-chain discovery bots: the number of active on-chain discovery bot-weeks falls only from 37 to 34. Instead, several extremely intensive bots disappear or downscale. Four on-chain discovery bots alone account for a large share of the drop from the week beginning August 3 to the week beginning August 10, 2025: \href{https://basescan.org/address/0xfbea32bcde4aa517680e4eb2bfad4a37d0eb2235}{\texttt{0xfbea32...eb2235}} falls from 0.896T spam gas to 0, \href{https://basescan.org/address/0xAa92f3EE3dd31aF92F71ce51aB9028CFC6B4d47c}{\texttt{0xaa92f3...b4d47c}} falls from 0.853T to 0, \href{https://basescan.org/address/0x6bab41136c65bd672f8adba457644a5ff7c0b196}{\texttt{0x6bab41...0b196}} falls from 0.811T to 0.204T, and \href{https://basescan.org/address/0x9a3A2F6DC28fa9AFC8E03B155d3585c79A623940}{\texttt{0x9a3a2f...23940}} falls from 0.312T to 0. Combined with the selection effect of Flashblocks against high scan footprint transactions, this combination explains why aggregate spam gas falls sharply even though spam transaction counts remain large.


\end{document}